\documentclass[a4paper,UKenglish,cleveref,autoref,thm-restate]{lipics-v2021}
\crefname{section}{§}{§§}
\Crefname{section}{§}{§§}

\newcommand{\ToolName}{\textsc{Squid}}
\newcommand{\ToolNameNRR}{\textsc{SquidNrr}}
\newcommand{\ToolNameNQR}{\textsc{SquidNbr}}

\newcommand{\CodeQL}{\textsc{CodeQL}}

\newcommand{\Continue}{\textnormal{\textbf{continue}}}
\newcommand{\Break}{\textnormal{\textbf{break}}}

\newcommand{\CaseNumber}{31}

\usepackage{pifont}
\usepackage{tcolorbox}
\usepackage[linesnumbered,commentsnumbered,ruled,vlined]{algorithm2e}
\usepackage{hyperref}
\SetAlgoSkip{}

\usepackage{amsthm}
\usepackage{balance}

\theoremstyle{definition}

\newtheorem{thm3}{Definition}[section]
\newtheorem{newdefinition}[thm3]{Definition}
\newtheorem{thm4}{Example}[section]
\newtheorem{newexample}[thm4]{Example}
\newtheorem{thm5}{Theorem}[section]
\newtheorem{ntheorem}[thm5]{Theorem}
\newtheorem{thm6}{Property}[section]
\newtheorem{nproperty}[thm6]{Property}

\newcommand{\mybox}[1]{
	\begin{tcolorbox}[
		boxsep=-0.5pt,
		standard jigsaw,
		boxrule=0.6pt,
		opacityback=0,
		sharp corners]
		#1
	\end{tcolorbox}
}

\SetAlCapSty{xAlCapSty}

\newcommand{\mytodoblue}[1]{\textcolor{blue}{\ding{46}~{\sf}~#1}}
\newcommand{\mytodored}[1]{\textcolor{red}{\ding{46}~{\sf}~#1}}

\newcommand{\mytodoorange}[1]{\textcolor{orange}{\ding{46}~{\sf}~#1}}
\newcommand{\mytodocyan}[1]{\textcolor{cyan}{\ding{46}~{\sf}~#1}}

\DeclareSymbolFont{extraup}{U}{zavm}{m}{n}
\DeclareMathSymbol{\varheart}{\mathalpha}{extraup}{86}
\DeclareMathSymbol{\vardiamond}{\mathalpha}{extraup}{87}
\DeclareMathSymbol{\varsuit}{\mathalpha}{extraup}{88}

\usepackage{listings}
\usepackage{microtype} 

\usepackage{xargs}                      

\usepackage{balance}
\usepackage{algpseudocode}
\usepackage{pifont}
\usepackage{mathtools}
\usepackage{amsmath}
\usepackage{url}
\usepackage{subcaption}

\usepackage{tabulary}
\usepackage{etoolbox}
\usepackage{mathrsfs}  
\usepackage{booktabs}

\usepackage[T1]{fontenc}

\usepackage{semantic}
\usepackage{url}

\usepackage{enumitem}
\usepackage{balance}

\usepackage{bm}

\usepackage{pgf,tikz}

\usepackage[utf8]{inputenc}
\usepackage{fmtcount}
\usepackage{threeparttable}

\usepackage{cleveref}

\newif\ifshowcomments
\showcommentstrue

\ifshowcomments
\newcommand{\wcp}[1]{\mytodocyan{[wcp: #1]}}
\newcommand{\yao}[1]{\mytodored{[yao: #1]}}
\newcommand{\tang}[1]{\mytodoorange{[tang: #1]}}
\newcommand{\shi}[1]{\mytodoblue{[shi: #1]}}
\newcommand{\wwy}[1]{\mytodoblue{[wwy: #1]}}
\newcommand{\cpw}[1]{\mytodocyan{[wcp: #1]}}
\else
\newcommand{\wcp}[1]{}
\newcommand{\yao}[1]{}
\newcommand{\tang}[1]{}
\newcommand{\shi}[1]{}
\newcommand{\wwy}[1]{}
\newcommand{\cpw}[1]{}
\fi

\title{Synthesizing Conjunctive Queries for Code Search} 

\author{Chengpeng Wang}{The Hong Kong University of Science and Technology, China}{cwangch@cse.ust.hk}{https://orcid.org/0000-0003-0617-5322}{}
\author{Peisen Yao}{Zhejiang University, Hangzhou, China}{pyaoaa@zju.edu.cn}{https://orcid.org/0000-0003-0342-9518}{}
\author{Wensheng Tang}{The Hong Kong University of Science and Technology, China}{wtangae@cse.ust.hk}{https://orcid.org/0000-0002-4259-3321}{}
\author{Gang Fan}{Ant Group, Shenzhen, China}{fangang@antgroup.com}{https://orcid.org/0000-0002-8633-6036}{}
\author{Charles Zhang}{The Hong Kong University of Science and Technology, China}{charlesz@cse.ust.hk}{https://orcid.org/0000-0001-6417-1034}{}

\authorrunning{C. Wang, P. Yao, W. Tang, G. Fan, and C. Zhang} 

\Copyright{Chengpeng Wang, Peisen Yao, Wensheng Tang, Gang Fan, and Charles Zhang} 

\ccsdesc[500]{Software and its engineering~Automatic programming}
\ccsdesc[500]{Human-centered computing~User interface programming}

\keywords{Query Synthesis, Multi-modal Program Synthesis, Code Search} 

\acknowledgements{We thank the anonymous reviewers, Xiao Xiao, and Xiaoheng Xie for their helpful comments. 
	The authors are supported by the RGC16206517,
	ITS/440/18FP and PRP/004/21FX grants from the Hong Kong
	Research Grant Council and the Innovation and Technology
	Commission, Ant Group, and the donations from Microsoft
	and Huawei. Peisen Yao is the corresponding author.}

\nolinenumbers 

\EventEditors{John Q. Open and Joan R. Access}
\EventNoEds{2}
\EventLongTitle{36th European Conference on Object-Oriented Programming (ECOOP 2023)}
\EventShortTitle{ECOOP 2023}
\EventAcronym{ECOOP}
\EventYear{2023}

\EventDate{July 17-21, 2023}
\EventLocation{Seattle}
\EventLogo{}
\SeriesVolume{263}
\ArticleNo{XX}

\begin{document}

\maketitle

\begin{abstract}
This paper presents \ToolName, a new conjunctive query synthesis algorithm for searching code with target patterns. 
Given positive and negative examples along with a natural language description, \ToolName\ analyzes the relations derived from the examples by a Datalog-based program analyzer and synthesizes a conjunctive query expressing the search intent. The synthesized query can be further used to search for desired grammatical constructs in the editor. To achieve high efficiency, we prune the huge search space by removing unnecessary relations and enumerating query candidates via refinement. We also introduce two quantitative metrics for query prioritization to select the queries from multiple candidates, yielding desired queries for code search. We have evaluated \ToolName\ on over thirty code search tasks. It is shown that \ToolName\ successfully synthesizes the conjunctive queries for all the tasks, taking only 2.56 seconds on average.
\end{abstract}

\section{Introduction}
\label{sec:intro}


Developers often need to search their code for target patterns in various scenarios,
such as API understanding~\cite{LiWWYXM16}, code refactoring~\cite{YangSLYC18}, and program repair~\cite{TianR17}.
According to recent studies~\cite{NaikMSWNR21, LiuXLGYG22},
existing efforts have to
compromise between ease of use and capability.
Most mainstream IDEs~\cite{IntelliJ} only support string match or structural search of restrictive grammatical constructs
although complex user interactions are not required.
Besides, static program analyzers, such as Datalog-based program analyzers~\cite{SmaragdakisB10, naik2011chord, AvgustinovMJS16},
provide deep program facts for users to explore advanced patterns,
while users have to customize the analyzers to meet their needs~\cite{ChristakisB16}.
For example, if users want to explore code patterns with the Datalog-based program analyzer \CodeQL~\cite{AvgustinovMJS16},
they have to learn the query language
to access the derived relational representation.
However, there always exists a non-trivial gap between a user's search intent and a customized query.
A large number of complex relations make query writing involve strenuous efforts, especially in formalizing search intents and debugging queries, which hinders the usability of \CodeQL\ for code search.

\smallskip
\emph{\textbf{Our Goal.}}
We aim to propose a query synthesizer to unleash the power of a Datalog-based program analyzer
for code search.
To show the search intent,
a user can specify a synthesis specification consisting of positive examples, negative examples, 
and a natural language description.
Specifically,
positive and negative examples indicate desired and non-desired grammatical constructs, respectively,
while the natural language description shows the search intent by a sentence.
Our synthesizer is expected to generate a query separating positive examples from negative ones,
which can support code search in the editor.
In this work, we focus on conjunctive queries,
which have been recognized as queries of an important form to support search tasks~\cite{GottlobKS06}.


\smallskip
Consider a usage scenario: 
Find all the methods receiving a parameter with \textsf{Log4jUtils} type and giving a return with \textsf{CacheConfig} type.
The user can provide the synthesis specification shown in Figure~\ref{fig:formulation_ex2}(a).
With the relational representation in Figure~\ref{fig:formulation_ex2}(b) derived from the examples,
our synthesizer would synthesize the conjunctive query in Figure~\ref{fig:formulation_ex2}(c) to express the search intent.
In particular, our synthesis specification is easy to provide.
The users can often copy desired grammatical constructs from an editor
as positive examples~\cite{NaikMSWNR21}
and then mutate them to form negative ones.
Meanwhile, they can express their need to search code with a brief sentence as the description.
Thus, an effective and efficient synthesizer 
enables the users to express the search intent from a high-level perspective,
serving as a user-friendly interface for code search.

\begin{figure}[t]
	\centering
	\includegraphics[width=0.85\linewidth]{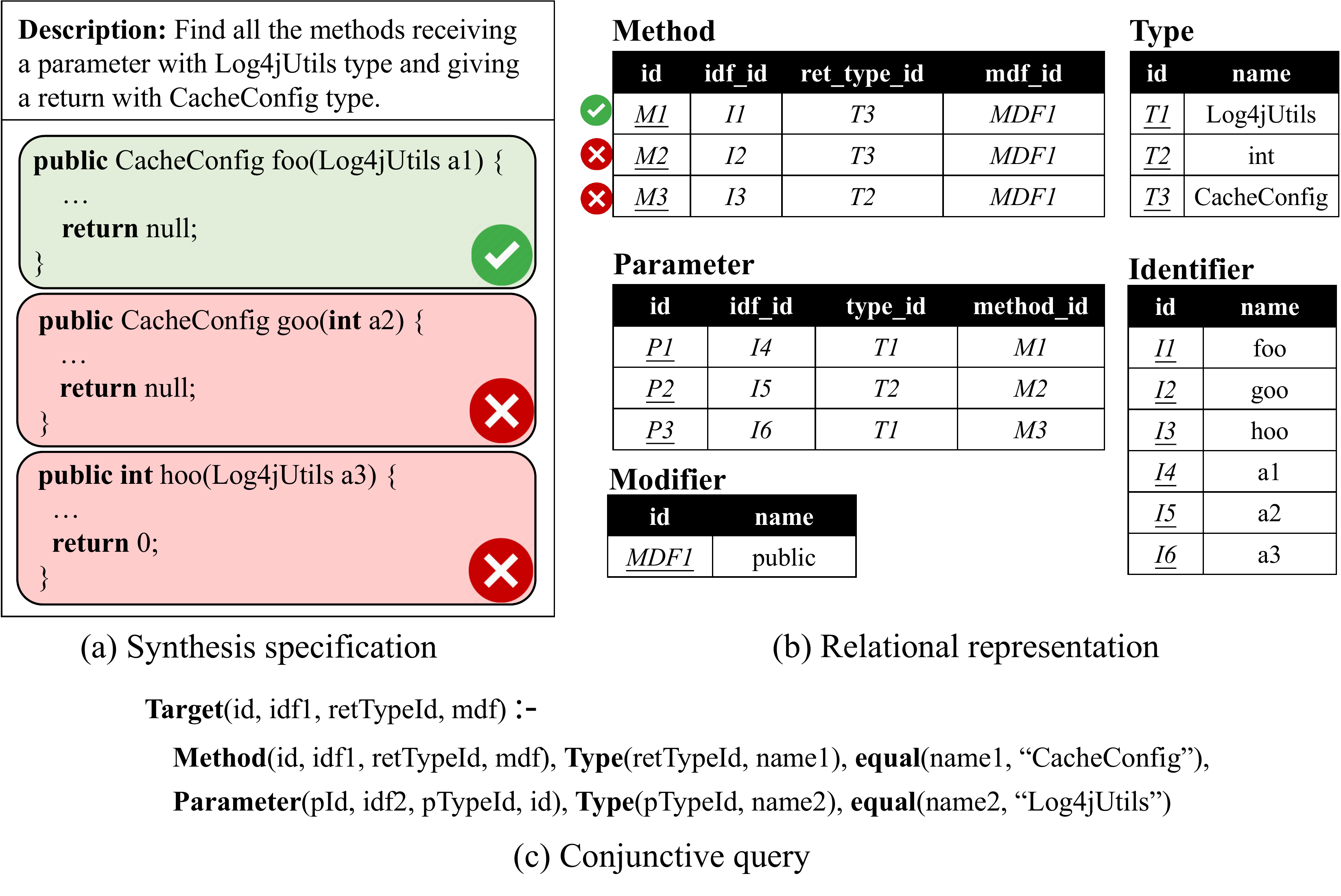}
	\vspace{-3mm}
	\caption{A motivating example\protect\footnotemark}   
	\label{fig:formulation_ex2}
	\vspace{-5mm}
\end{figure}
\footnotetext{We show five relations as examples, while a Datalog-based analyzer can derive over a hundred relations.}



\smallskip
\emph{\textbf{Challenges.}}
Nevertheless,
it is far from trivial to synthesize a conjunctive query for code search.
First, a Datalog-based program analyzer can generate many relations with multiple attributes
as the relational representation.
For example, \CodeQL\ exposes over a hundred relations to users for query writing~\cite{CodeQL}.
The various choices of selecting relations and enforcing conditions on attributes
induce a dramatically huge search space in the synthesis,
which can involve both the comparisons between attributes and string constraints,
posing a significant challenge to achieving high efficiency.
Second, there often exist multiple query candidates that separate positive examples from negative ones,
while several candidates can suffer from the over-fitting problem,
failing to express the search intent with no bias~\cite{XiongW22}.
An ineffective query candidate selection would mislead the synthesizer into returning wrong queries and 
further cause the failure of code search.

\smallskip
\emph{\textbf{Existing Effort.}}
There are three major lines of existing effort.
The first line of the studies utilizes input-output examples to synthesize queries in various forms, such as analytic SQL queries~\cite{ZhouBCW22, FengMGDC17} and relational queries~\cite{SiLZAKN18, ThakkarNSANR21}.
Although the queries often have expressive syntax, the synthesizers only take a few relations as input, not facing hundreds of relations as ours.
The second line of approaches is the component-based synthesis technique~\cite{FengM0DR17, PerelmanGBG12}, which leverages type signatures to enumerate well-typed programs.
However, a significant number of comparable attribute pairs still induce an explosively huge search space even if we adopt the techniques by guiding the search with the schema.
The third line of the studies derives program sketches from natural language descriptions via semantic parsing~\cite{LeiLBR13} and prioritizes feasible solutions with probability models~\cite{Yaghmazadeh0DD17, BaikJCJ20}.
Unfortunately, the ambiguity of natural languages and the inadequacy of the training process can make a semantic parser ineffective and eventually miss optimal solutions~\cite{RazaGM15}.
It is also worth noting that existing techniques do not attempt to select a feasible solution that maximizes or minimizes a specific metric.
Although several inductive logic learning-based techniques adopt heuristic priority functions
to accelerate the synthesis~\cite{ThakkarNSANR21},
they do not guarantee the optimality of the synthesized queries,
and thus can not resolve the query candidate selection in our problem.

\smallskip
\emph{\textbf{Our Solution.}}
Our key idea comes from three critical observations.
First, only a few relations contribute to separating positive examples from negative ones.
For example, the methods in Figure~\ref{fig:formulation_ex2}(a) have the same modifier,
indicating that the relation \textsf{Modifier} is unnecessary.
Second, adding an attribute comparison expression or a string constraint
to the condition of a conjunctive query yields
a stronger restriction on grammatical constructs.
If a query misses a positive example,
we cannot obtain a query candidate by strengthening the query.
Third, a desired query tends to constrain grammatical constructs
mentioned in the natural language description sufficiently.
For the instance in Figure~\ref{fig:formulation_ex2},
the query extracting the methods with the return type \textsf{CacheConfig} is a query candidate but not a desired one,
as it does not pose any restriction on parameters.

Based on the observations, we realize that it is possible to narrow down necessary relations and avoid their infeasible compositions
to prune the search space,
and meanwhile, select query candidates with the guidance of the natural language description.
According to the insight, we present a multi-modal synthesis algorithm \ToolName\ with three stages:
\begin{itemize}[leftmargin=*]
\item To narrow down the relations, we introduce the notion of the \emph{dummy relations} to depict the relations unnecessary for the synthesis and propose the \emph{representation reduction} to exclude dummy relations, which prunes the search space effectively.

\item To avoid infeasible compositions of relations, we perform the \emph{bounded refinement} to enumerate the queries,
skipping the unnecessary search for the queries that exclude a positive example.
Particularly, the string constraints are synthesized by computing the longest common substrings,
which is achieved efficiently in the refinement.

\item To select desired queries, we establish the dual quantitative metrics, namely \emph{named entity coverage} and \emph{structural complexity}, and select query candidates by optimizing them as the objectives, which creates more opportunities for returning desired queries.
\end{itemize}

We implement \ToolName\ and evaluate it on \CaseNumber\ code search tasks.
It successfully synthesizes desired queries for each task
in 2.56 seconds on average.
Besides, the representation reduction and the bounded refinement are crucial to its efficiency.
Skipping either of them would increase the average time cost to around 8 seconds, and several tasks cannot be finished within one minute.
Meanwhile, dual quantitative metrics play critical roles in the selection.
Applying only one metric would make 12 or 7 tasks fail due to the synthesized non-desired queries.
We also state and prove the soundness, completeness, and optimality of \ToolName.
If there exist query candidates for a given synthesis specification,
\ToolName\ always returns query candidates optimizing two proposed metrics to express the search intent.
To summarize, our work makes the following key contributions:
\begin{itemize}[leftmargin=*]
\item We propose a multi-modal conjunctive query synthesis problem. An effective and efficient solution can serve as a user-friendly interface of a Datalog-based analyzer for code search.

\item We design an efficient algorithm \ToolName\ for an instance of our synthesis problem, which automates the code search tasks in real-world scenarios.

\item We implement \ToolName\ as a tool and evaluate it upon \CaseNumber\ code search tasks, showing that \ToolName\ synthesizes the desired queries successfully and efficiently.
\end{itemize}

\section{Overview}

This section demonstrates a motivating example and briefs the key idea of our approach.

\subsection{Motivating Example}
Suppose a developer wants to avoid the security issue caused by \textsf{log4j} library~\cite{Log4j}.
He or she may examine the methods that receive a \textsf{Log4jUtils} object as a parameter and return a \textsf{CacheConfig} object.
One choice is to leverage the built-in search tools of the IDEs to search the code lines containing \textsf{Log4jUtils} or \textsf{CacheConfig},
while the string matching-based search cannot filter grammatical constructs according to their kinds.
Although several IDEs enable the structural search~\cite{IntelliJ},
their non-extensible templates only support searching for grammatical constructs of restrictive kinds.
Another alternative is to write a query depicting the target pattern and evaluate it with a Datalog-based program analyzer, such as \textsc{CodeQL}~\cite{AvgustinovMJS16}.
However, it not only involves great laborious efforts in query language learning but also creates the burden of query writing and debugging.

\smallskip
To improve the usability and capability of code search,
we aim to synthesize a query for a Datalog-based program analyzer.
As shown in Figure~\ref{fig:formulation_ex2}(a), a user can specify the synthesis specification to indicate the search intent.
Specifically, the positive and negative examples show the desired and undesired grammatical constructs, respectively,
while the natural language description demonstrates the search intent in a sentence.
Based on a Datalog-based program analyzer, 
we can convert the examples to a set of relations as the relational representation
along with positive and negative tuples,
which are shown in Figure~\ref{fig:formulation_ex2}(b).
For example, the first tuple in the relation \textsf{Method} is the positive tuple
indicating the method \textsf{foo} in Figure~\ref{fig:formulation_ex2}(a), which is a positive example.
If we automatically synthesize the conjunctive query in Figure~\ref{fig:formulation_ex2}(c),
the user does not need to delve into the relations and, instead 
specifies the synthesis specification from a high-level perspective.

\subsection{Synthesizing Conjunctive Queries}
The query synthesizer should effectively generate the desired queries
that express the search intent correctly.
However, it is stunningly challenging to obtain an effective and efficient synthesizer.
First, we have to tackle a great number of the relations and their attributes when we choose relevant relations and enforce correct constraints upon them,
which can involve both comparisons over attributes and string constraints.
Second, the non-uniqueness of query candidates
creates the obstacle of selecting proper candidates.
Any improper selection would return a non-desired query, leading to code search failure.
To address the challenges,
we propose a new multi-modal synthesis algorithm \ToolName.
As shown in Figure~\ref{fig:workflow},
\ToolName\ consists of three phases, which come from the following three ideas.

\begin{figure}[t]
	\centering
	\includegraphics[width=0.75\linewidth]{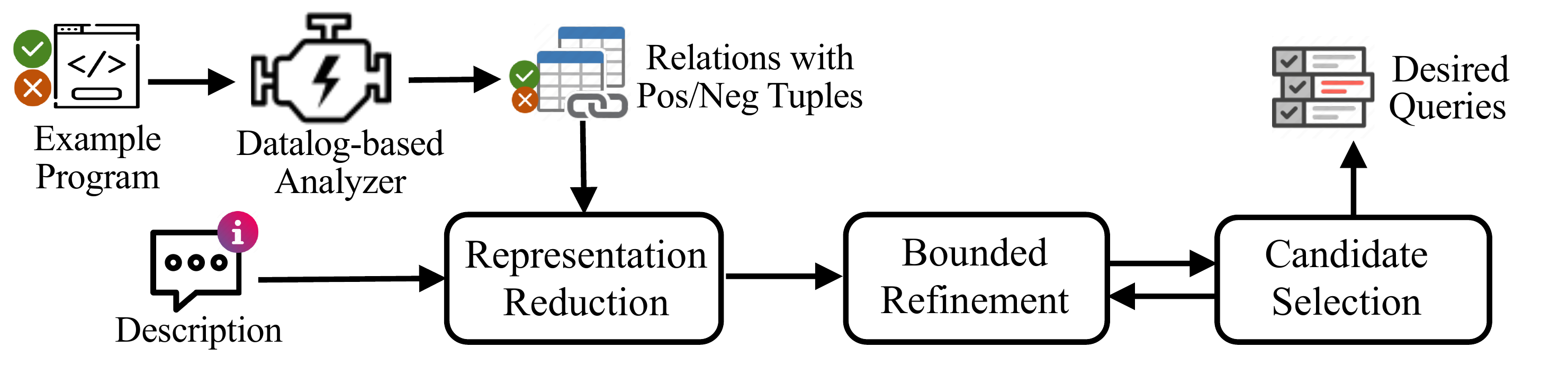}
   \vspace{-2mm}
	\caption{The overview of \ToolName}   
	\label{fig:workflow}
	\vspace{-5mm}
\end{figure}

\smallskip
\emph{\textbf{Idea 1: Removing dummy relations.}} 
Although there are many relations potentially used in the synthesis,
we can identify a class of relations, named \emph{dummy relations}, as unnecessary ones and then discard them safely.
Specifically, a relation is dummy if it cannot separate a positive tuple from a negative one.
As an example, the methods in Figure~\ref{fig:formulation_ex2}(a) have the same modifier, 
which is shown by the same values of the foreign keys \textsf{mdf\_id} of the relation \textsf{Method} in Figure~\ref{fig:formulation_ex2}(b).
This indicates that the relation \textsf{Modifier} has no impact on excluding negative tuples and thus can be discarded to prune the search space.
Based on this insight, we propose the \textbf{representation reduction} to remove the dummy relations,
narrowing down the necessary relations for the synthesis.

\smallskip
\emph{\textbf{Idea 2: Enumerating query candidates via refinement.}}
According to the query syntax,
the constraints, including attribute comparisons and string constraints,
pose restrictions on grammatical constructs.
This implies that
we cannot obtain a query candidate by refining the query that excludes a positive tuple.
In Figure~\ref{fig:formulation_ex2},
we may obtain a query that enforces both the parameter and the return value of a method have the same type.
Obviously, the query excludes the method \textsf{foo}, which is a positive tuple,
and thus, we should stop strengthening the restrictions on grammatical constructs.
Based on this insight, we introduce the technique of the \textbf{bounded refinement},
which adds conditions for the query enumeration and discards the queries that exclude any positive tuples.
Thus, we can avoid enumerating infeasible compositions of relations,
which further prunes the search space effectively.

\smallskip
\emph{\textbf{Idea 3: Dual quantitative metrics for selection.}}
Desired queries not only separate positive tuples from negative ones but also tend to cover as many program-related named entities as possible. 
In Figure~\ref{fig:formulation_ex2}, we may obtain a query candidate that restricts the return type to be a \textsf{CacheConfig} object.
However, it does not pose any restriction on the parameters and leaves the named entity \text{``}parameter\text{''} uncovered,
showing that it does not express the intent sufficiently.
Meanwhile, Occam's razor~\cite{BlumerEHW87} implies that desired queries should be as simple as possible.
Hence, we introduce the \emph{named entity coverage} and the \emph{structural complexity}
as the dual quantitative metrics,
and perform the \textbf{candidate selection} to identify desired queries
by optimizing the metrics.
Finally, we blend the selection with the refinement and terminate the enumeration when unexplored candidates cannot be better than the current ones, 
further avoiding the unnecessary enumerative search.
\section{Problem Formulation}
This section first presents the program relational representation (\cref{subsec:prr}) 
and then introduces the conjunctive queries for code search (\cref{subsec:cq}).
Lastly, we state the multi-modal conjunctive query synthesis problem
and brief the roadmap of technical sections (\cref{subsec:ps}).

\subsection{Program Relational Representation}
\label{subsec:prr}


First of all, we formally define the concept of the \emph{relation} as the preliminary.

\begin{newdefinition}(Relation)
	\label{def:relation}
	A relation $R(a_1, \cdots, a_n)$ is a set of tuples $(t_1, \cdots, t_n)$, where $n$ is the arity of $R$.
	For each $1 \leq i \leq n$, $a_i$ is the attribute of the relation.
\end{newdefinition}

A relation is structured data that stores the details of different aspects of an entity.
Concretely, a Datalog-based program analyzer encodes the program properties with a set of relations in a specific schema~\cite{AvgustinovMJS16, NaikMSWNR21, PashakhanlooN0D22}.
In what follows,
we define the \emph{relational representation} of a program.

\begin{newdefinition}(Relational Representation)
	\label{def:pf}
	Given a program, its relational representation $\mathcal{R}$ is a set of relations over the following schema $\Gamma$.
	Specifically, $\Gamma$ maps a relation symbol $R$ to an $n$-tuple of pairs, where each element in the tuple $\Gamma(R)$ has one of the following form:
	\begin{itemize}[leftmargin=*]
		\item $(id, R)$: The attribute $id$ is the primary key of the relation $R$.
		\item $(a, R')$: The foreign key $a$ in the relation $R$, referencing the primary key of $R' \in dom(\Gamma)$.
		\item $(a, \textsf{STR})$: The attribute $a$ has a string value indicating the textual information.
	\end{itemize}
	Particularly, we say $\Gamma$ as the language schema.
	Without introducing the ambiguity, we use $(a, \cdot) \in \Gamma(R)$ to indicate that $(a, \cdot)$ is an element of the tuple $\Gamma(R)$.
\end{newdefinition}

\begin{newexample}
Figure~\ref{fig:formulation_ex2}(b) shows five relations as examples,
where the values of the primary keys are italic and underlined, and the foreign keys are italic.	
Based on the first tuple in the relation \textsf{Method}, we can track the identifier, the return type, and the modifier of the method \textsf{foo} based on the foreign keys.
Similarly, we can identify the identifier and the type of each parameter based on the relation \textsf{Parameter}.
\end{newexample}

Essentially, the relational representation of a program encodes the program properties
with relations,
which depicts the relationship of grammatical constructs in the program.
In reality, various relations can be derived with Datalog-based program analyzers,
such as \textsf{ReferenceType} and \textsf{VarPointsTo} provided by \textsc{Doop}~\cite{SmaragdakisB10},
depicting the type information and points-to facts, respectively.
Due to the space limit, we only show five relations in Figure~\ref{fig:formulation_ex2}(b).


\subsection{Conjunctive Queries}
\label{subsec:cq}
To simplify the presentation,
we formulate the conjunctive queries as the relational algebra expressions~\cite{AbiteboulHV95} in the rest of the paper.
In what follows, we first brief several relational algebra operations
and then introduce the conjunctive queries for code search.

\begin{newdefinition}(Relational Algebra Operations)
	In a relational algebra, the selection, projection, Cartesian product, and rename operations are defined as follows:
	\begin{itemize}[leftmargin=*]
		\item $\sigma_{\Theta}(R):=\{ (t_1, \ldots, t_n) \in R \ | \ \ [a_i \mapsto t_i \ | \ 1 \leq i \leq n] \models \Theta \}$ is 
		the selection of a relation $R$ with the selection condition $\Theta$ over its attributes.
		\item $\Pi_{\mathbf{a'}}(R):=\{ (\mathbf{t}.a_1', \dots, \mathbf{t}.a_k')  \ | \  \mathbf{t} \in R  \}$ is
		the projection of a relation $R$ upon a tuple of attributes $\mathbf{a'}$, denoted by $\Pi_{\mathbf{a'}}(R)$.
		Here $\mathbf{a'} = (a_1', \dots, a_k')$. $\mathbf{t}.a$ is the value of the attribute $a$ in the tuple $\mathbf{t}$.
		\item $R_1 \times R_2 := \{ (t_1^{1}, \ldots, t_{n_1}^{1}, t_1^{2}, \ldots, t_{n_2}^{2}) \ | \  (t_1^{1}, \ldots, t_{n_1}^{1}) \in R_1, (t_1^{2}, \ldots, t_{n_2}^{2}) \in R_2 \}$ is
		the Cartesian product of two relations $R_1$ and $R_2$.
		\item The rename of a relation $R$, denoted by $\rho_{A}(R)$, yields the same relation named $A$.
	\end{itemize}
\end{newdefinition}

The relational algebra operations enable us to manipulate the relational representation to search desired grammatical constructs.
Specifically, we often need to search specific grammatical constructs via string match and enforce them to satisfy several constraints simultaneously.
Now we formalize the conjunctive queries
to express the code search intent.

\begin{newdefinition}(Conjunctive Query)
	\label{def:cqcq}
	Given the relational representation $\mathcal{R}$,
	a conjunctive query $R_Q$ is a relational algebra expression of the form
	$\Pi_{(A_i.*)} (\sigma_{\Theta} (\rho_{A_1}(R_1) \times \ldots \times \rho_{A_m}(R_m)))$,
	where $R_i \in \mathcal{R}$, $\Theta := \phi_1 \land \cdots \land \phi_n$, 
	and each $\phi_i (1 \leq i \leq n)$ occurring in the selection condition $\Theta$ is an atomic condition in the following two forms:
	\begin{itemize}[leftmargin=*]
		\item An atomic equality formula $A_j.\textsf{id} = A_k.a$,
		where $a$ is the foreign key and $(a, R_j) \in \Gamma(R_k)$.
		\item A string constraint $p(A_k.a, \ell)$ over the string attribute $A_k.a$, 
		where $\ell$ is a string literal, and $p \in \{\textsf{equal}, \textsf{suffix}, \textsf{prefix}, \textsf{contain}\}$.
	\end{itemize}
	In particular, $\Pi_{(A_i.*)}$ indicates the projection upon all the attributes of the relation $A_i$.
\end{newdefinition}


The form of the conjunctive queries in Definition~\ref{def:cqcq} depicts the user intent
from two aspects.
First,
the atomic equality formulas encode the relationship between the grammatical constructs.
Second,
the four string predicates support the common scenarios of string matching-based code search.
We do not focus on synthesizing more expressive string constraints,
which is the orthogonal direction of program synthesis~\cite{Chen0YDD20, LeeSO16, PanHXD19}.
Based on the conjunctive queries,
we can simultaneously perform the string matching-based search and
filter the constructs with various relations in the program relational representation.


\begin{newexample}
	\label{ex:cq}
	We can formalize the query in Figure~\ref{fig:formulation_ex2}(c) as the relational algebra expression:
	$$\Pi_{(A_1.*)} (\sigma_{\Theta} (\rho_{A_1}(\textsf{Method}) \times \rho_{A_2}(\textsf{Type}) \times \rho_{A_3}(\textsf{Parameter}) \times \rho_{A_4}(\textsf{Type})))$$
	where the selection condition $\Theta$ in the selection operation is as follows:
	\begin{align*}
		\Theta := &(A_1.\textsf{id} = A_3.\textsf{method\_id}) \land (A_1.\textsf{ret\_type\_id} = A_2.\textsf{id}) \land (A_3.\textsf{type\_id} = A_4.\textsf{id}) \land \\
		& \textsf{equal}(A_2.\textsf{name}, ``\textsf{CacheConfig}\text{''}) \land \textsf{equal}(A_4.\textsf{name}, ``\textsf{Log4jUtils}\text{''})
	\end{align*}
\end{newexample}

The conjunctive queries can be instantiated with various flavors~\cite{GottlobKS06, AbiteboulHV95, ChandraM77}.
The \emph{select-from-where queries} are the instantiations of the conjunctive queries in SQL.
Besides, a simple Datalog program can also express the conjunctive query with a single Datalog rule.
In our paper, we formulate a conjunctive query as a relational algebra expression.
Our implementation actually synthesize the conjunctive queries as Datalog programs,
which are evaluated by a Datalog solver over the program relational representation for code search.

\subsection{Multi-modal Conjunctive Query Synthesis Problem}
\label{subsec:ps}

To generate the conjunctive query for code search,
the users need to specify their intent as the specification.
In our work, we follow the premise of the recent studies on the multi-modal program synthesis~\cite{Chen0YDD20, BaikJCJ20}
that multiple modalities of information
can go arm in arm with each other, serving as the informative specification for the synthesis. 
Specifically, the users provide a natural language sentence
to describe the target code pattern and
provide several grammatical constructs as positive and negative examples.
Notably, the multi-modal synthesis specification is easy to provide.
When the users want to explore a specific code pattern,
they can describe the pattern briefly in a natural language
and provide several examples from their editors instead of
delving into the details of the underlying relations and their attributes, 
enabling users to generate a query for code search in a declarative manner.

Based on the multi-modal synthesis specification,
the positive and negative examples 
are converted to the tuples in a specific relation.
For example, Figure~\ref{fig:formulation_ex2} shows a set of relations as the relational representation of 
the examples.
To formalize our problem better,
we define the notion of the \emph{relation partition} as follows.

\begin{newdefinition}(Relation Partition)
The relation partition of $R^{*} \in \mathcal{R}$ is a pair of two relations $(R_p^{*}, R_n^{*})$ satisfying
$R^{*} = R_p^{*} \cup R_n^{*}$ and $R_p^{*} \cap R_n^{*} = \emptyset$.
The relation partition is non-trivial if and only if $R_p^{*} \neq \emptyset$ and $R_n^{*} \neq \emptyset$.
We say the tuples in $R_p^{*}$ and $R_n^{*}$ are positive and negative tuples, respectively.
\end{newdefinition}

\begin{newexample}
	\label{ex:5}
	As shown in Figure~\ref{fig:formulation_ex2}, 
	we can construct the relation partition $(R_p^{*}, R_n^{*})$, where $R_p^{*}=\{ (\textsf{M1}, \ \textsf{I1}, \ \textsf{T3}, \ \textsf{MDF1}) \}$
	and $R_n^{*} = \{ (\textsf{M2}, \ \textsf{I2}, \ \textsf{T3}, \ \textsf{MDF1}), (\textsf{M3}, \ \textsf{I3}, \ \textsf{T2}, \ \textsf{MDF1}) \}$.
	Obviously, $R_p^{*}$ and $R_n^{*}$ are disjoint,
	 and $R_p^{*} \cup R_n^{*}$ is exactly the relation $\textsf{Method}$.
\end{newexample}

The positive and negative tuples essentially depict the positive and negative examples, respectively.
Based on the program's relational representation,
the examples specified by the users can determine the positive and negative tuples,
which can be achieved in various manners.
Specifically, users can select a grammatical construct in the IDEs  or use a code sample in a specific coding standard~\cite{wilson2000java} as a positive example
and remove several sub-patterns from a positive example by mutation to construct negative ones.
Such positive and negative examples further constitute a sample code snippet, 
from which a Datalog-based analyzer derives a set of relations as the relational representation.
In this paper, we omit the details of positive/negative tuple generation
and formulate a \emph{multi-modal conjunctive query synthesis (MMCQS) problem} as follows.

\mybox{
	Given a relational representation $\mathcal{R}$, 
	a relation partition $(R_p^{*}, R_n^{*})$ of $R^{*} \in \mathcal{R}$,
	and a natural language description $s$,
	we aim to synthesize a conjunctive query $R_Q$ containing the positive tuples in $R_p^{*}$ and excluding the negative tuples in $R_n^{*}$.
}

\begin{newexample}
	\label{ex:problem}
	Figure~\ref{fig:formulation_ex2}(a) shows the multi-modal synthesis specification,
	which consists of a positive example, two negative examples,
	and a natural language description $s$
	as \text{``}Find all the methods receiving a \textsf{Log4jUtils}-type parameter and giving a \textsf{CacheConfig}-type return\text{''}.
	Leveraging a Datalog-based analyzer,
	we obtain the relational representation and the relation partition of \textsf{Method},
	which are shown in Figure~\ref{fig:formulation_ex2}(b).
	To automate the code search,
	we expect to synthesize the query in Fig~\ref{fig:formulation_ex2}(c) or Example~\ref{ex:cq}
	according to the relational representation, the relation partition, and the natural language sentence.
\end{newexample}

To promote the code search,
we propose an efficient synthesis algorithm \ToolName\ for the MMCQS problem, which is our main technical contribution.
As explained in~\cref{sec:intro},
it is challenging to solve the MMCQS problem efficiently,
which involves mitigating the huge search space and selecting queries from multiple candidates. 
In the following two sections,
we formalize the conjunctive query synthesis from a graph perspective (\cref{sec:sg}),
and illustrate the technical details of our synthesis algorithm \ToolName\ (\cref{sec:syn}),
which prunes the search space and selects desired queries effectively.

\section{Conjunctive Query Synthesis: A Graph Perspective}
\label{sec:sg}

This section presents a graph perspective of our conjunctive query synthesis problem.
Specifically, we introduce two graph representations of the language schema and the conjunctive queries, named the schema graph (\cref{subsec:sg}) and the query graph (\cref{subsec:gs}), respectively, which reduces the conjunctive query synthesis to the query graph enumeration.
Lastly, we summarize the section and highlight the technical challenges from a graph perspective (\cref{subsec:sgsummary}).


\subsection{Schema Graph}
\label{subsec:sg}
According to Definition~\ref{def:pf},
a relation in the language schema has three kinds of attributes,
namely a unique primary key, foreign keys, and string attributes.
Obviously, 
the selection condition $\Theta$ in $R_Q$ should only compare the foreign key of a relation with its referenced primary key or constrain the string attributes with string predicates.
To depict the possible ways of constraining the attributes,
we define the concept of the \emph{schema graph} as follows.

\begin{newdefinition}(Schema Graph)
	\label{def:sg}
	The schema graph $G_{\Gamma}$ of a language schema $\Gamma$ is $(N_{\Gamma}, E_{\Gamma})$:
	\begin{itemize}[leftmargin=*]
		\item The set $N_{\Gamma}$ contains the relation symbols in the schema or the string type \textsf{STR} as the nodes of the schema graph, i.e., $N_{\Gamma}:=dom(\Gamma) \cup \{ \textsf{STR} \}$.
		\item The set $E_{\Gamma}$ contains an edge $(n_1, n_2, a)$ if and only if either of the conditions holds:
		\begin{itemize}[leftmargin=*]
			\item $n_1, n_2 \in dom(\Gamma)$ and $(a, n_2)\in \Gamma(n_1)$: The relation $n_1$ has a foreign key named $a$ referencing the primary key of the relation $n_2$.
			\item $n_1 \in dom(\Gamma)$ and $(a, \textsf{STR}) \in \Gamma(n_1)$: $a$ is the string attribute of the relation $n_1$.
		\end{itemize}
	\end{itemize}
\end{newdefinition}

\begin{newexample}
	Consider the relations in the relational representation shown in Figure~\ref{fig:formulation_ex2}(b).
	We can construct the schema graph in Figure~\ref{fig:schema_graph}(a).
	The edge from \textsf{Method} to \textsf{Modifier} labeled with \textsf{mdf\_id}
	shows that the attribute \textsf{mdf\_id} of \textsf{Method} is a foreign key referencing the primary key of \textsf{Modifier}.
	Similarly, the edge from \textsf{Type} to \textsf{STR} labeled with \textsf{name} shows that the attribute \textsf{name} in \textsf{Type} is a string attribute.
\end{newexample}

Noting that there can exist multiple edges with different labels between two nodes in the schema graph,
which indicate that a relation take multiple foreign keys referencing the same relation or string attributes as the attributes.
Essentially, the schema graph encodes the available relations with its node set and 
depicts the valid forms of the atomic formulas appearing in the selection condition with its edge set.
Although we can compare the attributes of any relations flexibly,
a solution to our problem must take the valid form of the atomic formulas as its selection condition,
comparing the foreign keys with the referenced primary keys or examining the string attributes of the relations appearing in a Cartesian product.

\subsection{Query Graph}
\label{subsec:gs}

As formulated in Definition~\ref{def:cqcq},
there are two key components in the conjunctive query,
namely the Cartesian product and the selection condition.
Leveraging the schema graph,
we can represent the components with nodes and edges on the graph,
which uniquely determines a conjunctive query.
Formally, we introduce the notion of the \emph{query graph} as follows.

\begin{newdefinition}(Query Graph)
\label{def:qg}
Given a conjunctive query $Q$, its query graph $G_Q$ is $(N_Q, E_Q, \Phi_Q)$:
\begin{itemize}[leftmargin=*]
\item The set $N_Q$ contains $(R_i, A_i)$ or $(p, \ell)$ as a node in the query graph.
$R_i \in \mathcal{R}$ is a relation and $A_i$ is the unique relation identifier.
$p$ and $\ell$ are the string predicate and literal, respectively.
\item The set $E_Q$ contains $(n_1, n_2, a)$ as an edge,
corresponding to the equality atomic formula $A_j.a = A_k.\textsf{id}$ in the selection condition,
where $n_1 = (R_j, A_j)$ and $n_2 = (R_k, A_k)$.
\item The mapping $\Phi_Q$ maps a 3-tuple $(R_j, A_j, a)$ to a node $(p, \ell)$,
indicating the edge from $(R_j, A_j)$ to $(p, \ell)$ with the label $a$,
which corresponds to the string constraint $p(A_j.a, \ell)$.
\end{itemize}
\end{newdefinition}

\begin{figure}[t]
	\centering
	\includegraphics[width=0.9\linewidth]{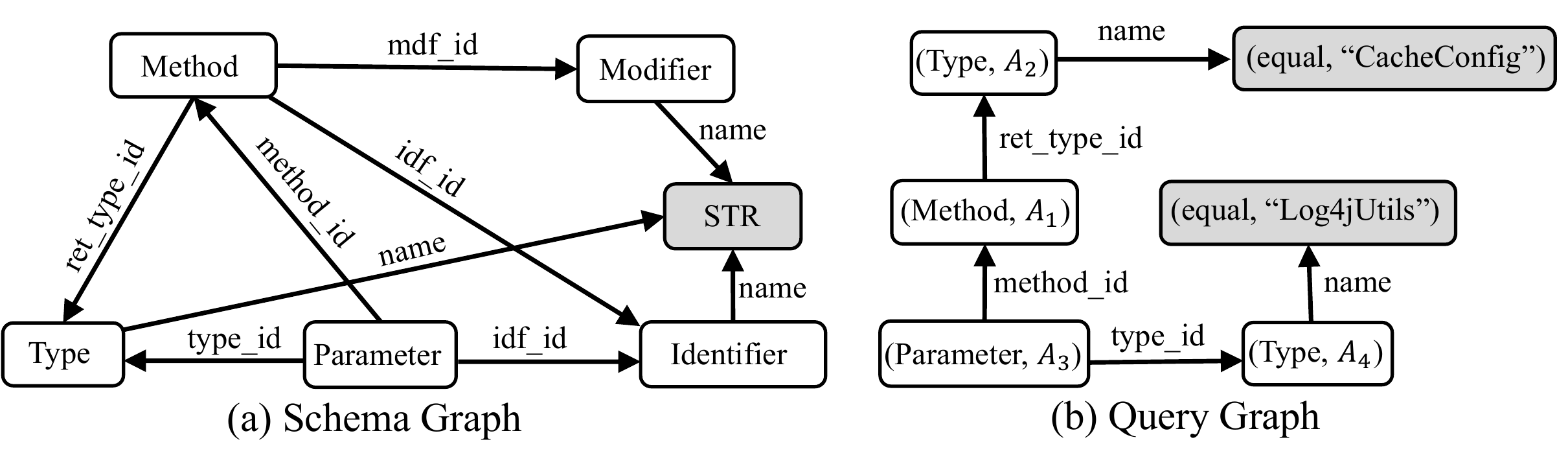}
	\vspace{-3mm}
	\caption{The examples of the schema graph and the query graph
	}   
	\vspace{-3mm}
	\label{fig:schema_graph}
\end{figure}

\begin{newexample}	
	Figure~\ref{fig:schema_graph}(b) shows an query graph of $R_Q$ in Example~\ref{ex:cq}.
	The white nodes show the four relations appearing in the Cartesian product,
	while the gray nodes indicate the string predicates and literals in the selection condition.
	The edges depict two kinds of atomic conditions.
	For example, the edge from $(\textsf{Parameter}, A_3)$ to $(\textsf{Method}, A_1)$ labeled with $\textsf{method\_id}$
	indicates the equality constraint $A_3.\textsf{method\_id}=A_1.\textsf{id}$.
	Meanwhile, the edge induced by $\Phi_Q(\textsf{Type}, \ A_4, \ \textsf{name}) = (\textsf{equal}, ``\textsf{Log4jUtils}\text{''})$
	indicates the string constraint $\textsf{equal}(A_4.\textsf{name}, ``\textsf{Log4jUtils}\text{''})$.
\end{newexample}

Essentially, a conjunctive query and a query graph are allotropes.
That is, there exists a bijection $\kappa$ mapping a conjunctive query $R_Q$ to a query graph $G_Q$ such that $G_Q = \kappa  (R_Q)$ and $R_Q = \kappa^{-1}(G_Q)$.
Thus, we can enumerate the conjunctive queries by enumerating the query graphs.
Besides, the schema graph restricts the form of the selection condition over the attributes.
If an edge with the label $a$ connects $(R_j, A_j)$ and $(R_k, A_k)$ in a query graph,
there should exist an edge labeled with $a$ connecting the relations $R_j$ and $R_k$ in the schema graph.
A similar argument also holds for the edge of the query graph indicated by the mapping $\Phi_{Q}$.
Therefore, we can reduce our conjunctive query synthesis to a search problem,
of which the search space is characterized as the set of query graphs.

\subsection{Summary}
\label{subsec:sgsummary}
Leveraging the schema graph, we have reduced the conjunctive query synthesis process to query graph enumeration. 
To obtain the desired queries, we only need to select the nodes and edges from the schema graph and create the node $(p, \ell)$ with proper string predicates and literals for constructing a query graph $G_Q$,
which should satisfy that the induced query $\kappa^{-1}(G_Q)$ separates positive tuples from negative ones.

Obtaining the desired queries with high efficiency is a non-trivial problem. The schema graph can be overwhelming, containing over a hundred nodes and edges, which leads to an enormous number of choices for relations occurring in the query graph. Even for a given set of relations, the flexibility of instantiating equality constraints and string constraints over their attributes can induce a large number of selection choices of edges in a query graph, which exacerbates the search space explosion problem. Additionally, the existence of multiple query candidates necessitates the effective selection of candidates, which is crucial for conducting a code search task. In the next section, we will detail our synthesis algorithm that addresses these challenges, resulting in high efficiency and effectiveness for code search.

\section{Synthesis Algorithm}
\label{sec:syn}

This section presents our synthesis algorithm \ToolName\ to solve the MMCQS problem.
\ToolName\ takes as input the relational representation $\mathcal{R}$ of an example program, a relation partition $(R_p^{*}, R_n^{*})$ of $R^{*} \in \mathcal{R}$, and a natural language description $s$.
It generates the query candidates separating positive tuples $\textbf{t}_p \in R_p^{*}$ from negative ones $\textbf{t}_n \in R_n^{*}$ ,
which can be further selected and then used for code search.
To address the challenges in~\cref{subsec:sgsummary},
\ToolName\ works with the following three stages:

\begin{itemize}[leftmargin=*]
	\item To tackle a large number of relations,
	\ToolName\ relies on the notion of dummy relations and conducts the representation reduction
	based on the positive and negative tuples,
	which effectively narrow down the relations that can appear in the query graph (\cref{sec:reduce}).
	
	\item To avoid unnecessary enumeration of edges in query graphs,
	we propose the bounded refinement by enumerating the query graphs based on the schema graph, 
	which essentially appends equality constraints and string constraints inductively (\cref{sec:refine}).
	
	\item To select query candidates, \ToolName\ identifies the named entities in the natural language description $s$ for the prioritization,
	and blends the selection with the refinement to collect the desired queries (\cref{sec:prior}).
\end{itemize}
We also formulate the soundness, completeness, and optimality of our algorithm (\cref{subsec:algsummary}).
For better illustration,
we use the synthesis instance shown in Figure~\ref{fig:formulation_ex2} throughout this section.

\subsection{Representation Reduction}
\label{sec:reduce}

To tackle the large search space,
we first propose the representation reduction to narrow down the relations possibly used in the query.
Specifically, we introduce the notion of the dummy relations to determine the characteristics of unnecessary relations (\cref{subsec:sndr}) and propose the algorithm of removing dummy relations for the representation reduction (\cref{subsec:dralg}).

\subsubsection{Dummy Relations}
\label{subsec:sndr}

There exists a class of relations, named \emph{dummy relations}, 
which cannot involve in distinguishing positive and negative tuples, such as the relation \textsf{Modifier} in Figure~\ref{fig:formulation_ex2}. 
Before defining them, we first introduce the \emph{undirected relation path} and the \emph{activated relation}.

\begin{newdefinition}(Undirected Relation Path)
	An undirected relation path from $R_{0}$ to $R_{k+1}$ in
	the schema graph $G_{\Gamma}=(N_{\Gamma}, E_{\Gamma})$
	 is $p: R_0 \hookrightarrow_{(a_0, d_0)} \cdots  \hookrightarrow_{(a_k, d_k)} R_{k + 1}$,
	where $R_i \in dom(\Gamma)$.
	Here, $d_i = 1$ if and only if $(R_i, R_{i+1}, a_i) \in E_{\Gamma}$,
	and $d_i = -1$ if and only if $(R_{i+1}, R_i, a_i) \in E_{\Gamma}$.
\end{newdefinition}

\begin{newdefinition}(Activated Relation)
\label{def:activated_relation}
Given a tuple $\mathbf{t}_0 \in R_{0}$ and
an undirected relation path $p: R_0 \hookrightarrow_{(a_0, d_0)} \cdots  \hookrightarrow_{(a_k, d_k)} R_{k + 1}$,
the activated relation of $\mathbf{t}_0$ along $p$ is 
$$\mathcal{I}(\mathbf{t}_0, p) = \{ \mathbf{t}_{k+1} \ | \ \mathbf{t}_{i+1} \in R_{i+1}, \ \textsf{ite}(d_i = 1, \ \textbf{t}_i.a_i = \mathbf{t}_{i+1}.\textsf{id}, \ \textbf{t}_i.id = \mathbf{t}_{i+1}.a_i), \ 0 \leq i \leq k \}$$
\end{newdefinition}
\begin{newexample}
	\label{ex:activated_relation}
	In in Figure~\ref{fig:formulation_ex2}, the path $p_1: \textsf{Method} \hookrightarrow_{(\textsf{method\_id}, -1)} \textsf{Parameter} \hookrightarrow_{(\textsf{type\_id}, \ 1)} 
	\textsf{Type}$,
	and $\mathbf{t}_p = (\textsf{M1}, \ \textsf{I1}, \ \textsf{T3}, \ \textsf{MDF1}) \in R_p^{*}$.
	We have 
	$\mathbf{t}_1 = (\textsf{P1}, \ \textsf{I4}, \ \textsf{T1}, \ \textsf{M1}) \in \textsf{Parameter}$, and $\mathbf{t}_2 = (\textsf{T1}, \textsf{Log4jUtils}) \in \textsf{Type}$,
	By inspecting other tuples, we have $\mathcal{I}(\mathbf{t}_p, p_1) = \{ (\textsf{T1}, \textsf{Log4jUtils}) \}$.
\end{newexample}

Intuitively, an undirected relation path can depict the restriction upon the relations in the Cartesian product of the query.
The activated relation actually contains the tuples 
enforcing the primary key of $\mathbf{t}_0$ to appear in a selected tuple.
Therefore, it is possible to narrow down the relations for the synthesis
by inspecting the activated relations of positive and negative tuples along each undirected relation path. According to the intuition, we formally introduce and define the notion of the \emph{dummy relations} as follows.

\begin{newdefinition}(Dummy Relation)
	\label{def:dummy_relation}
	Given a relation partition $(R_p^{*}, R_n^{*})$ of $R^{*} \in \mathcal{R}$,
	a relation $R \in \mathcal{R}$ is dummy if for every undirected relation path $p$ from $R^{*}$ to $R$, 
	either of the conditions is satisfied: 
	(1) There exists $\mathbf{t}_p \in R_p^{*}$ such that $\mathcal{I}(\mathbf{t}_p, p) = \emptyset$;
	(2) $\mathcal{I}(\mathbf{t}_p, p) = \mathcal{I}(\mathbf{t}_n , p)$ for any $\mathbf{t}_p \in R_p^{*}$ and  $\mathbf{t}_n \in R_n^{*}$.
\end{newdefinition}

Definition~\ref{def:dummy_relation} formalizes two characteristics of relations unnecessary  the synthesis.
First, the empty activated relation of a positive tuple $\mathbf{t}_p$ indicates the absence of the tuples in the relation $R$, making the tuple $\mathbf{t}_p$ appear.
Second, the relation $R$ cannot contribute to separating positive tuples from negative ones if several tuples in $R$ make all the positive and negative tuples appear simultaneously.
Thus, such relations can be discarded safely.

\begin{newexample}
	\label{ex:dummy_relation}
	Example~\ref{ex:activated_relation} shows that
	$\mathcal{I}(\mathbf{t}_p, p_1) \neq \emptyset$.
	Similarly,  $\mathcal{I}(\mathbf{t}_n, p_1) = \{ (\textsf{T2}, \textsf{int}) \} \neq \mathcal{I}(\mathbf{t}_p, p_1)$
	when $\mathbf{t}_n = (\textsf{M2}, \ \textsf{I2}, \ \textsf{T4}, \ \textsf{MDF1}) \in R_n^{*}$.
	Hence, $\textsf{Type}$ is not dummy.
	Besides, any undirected relation path $p_2$ from $\textsf{Method}$ to $\textsf{Modifier}$
	has the form
	$$\textsf{Method} [\hookrightarrow_{(\textsf{mdf\_id}, +1)} \textsf{Modifier} \hookrightarrow_{(\textsf{mdf\_id}, -1)} \textsf{Method}]^{*} \hookrightarrow_{(\textsf{mdf\_id}, +1)} \textsf{Modifier}$$
	where $[\cdot]^{*}$ indicates the repeated subpath.
	We find that
	$\mathcal{I}(\mathbf{t}_p, p_2) = \mathcal{I}(\mathbf{t}_n, p_2) = \{  \textsf{MDF1}, \textsf{public}   \}$
	for any $\mathbf{t}_p \in R_p^{*}$ and $\mathbf{t}_n \in R_n^{*}$.
	Thus, the relation $\textsf{Modifier}$ is a dummy relation.
\end{newexample}

\subsubsection{Removing Dummy Relations}
\label{subsec:dralg}

Based on Definition~\ref{def:dummy_relation},
identifying dummy relations involves two technical parts.
First, we should collect all the undirected relation paths from $R^{*}$ to each relation.
Second, we need to compute $\mathcal{I}(\mathbf{t}_p, p)$ and $\mathcal{I}(\mathbf{t}_n, p)$ 
for each undirected relation path $p$ and positive/negative tuple.
However, the schema graph can contain a large and even infinite number of undirected relation paths from $R^{*}$.
Any cycle induces the infinity of the path number,
making it tricky to examine the conditions in Definition~\ref{def:dummy_relation}.
Fortunately,
we realize that the number of cycles in a path does not affect the activated relation, 
which is stated in the following property.

\begin{nproperty}
	\label{property}
	Given any $\mathbf{t}_0 \in R_{\textsf{0}}$, we have $\mathcal{I}(\mathbf{t}_0, p) = \mathcal{I}(\mathbf{t}_0, p')$ for $p$ and $p'$ as follows:
	\begin{align*}
		p &: R_{\textsf{0}} \hookrightarrow_{(a_0, d_0)} \cdots R_l \ [\hookrightarrow_{(a_l', d_l')} \cdots R_{l+t} \hookrightarrow_{(a_{l+t}', d_{l+t}')}  R_l]^{+} \ \hookrightarrow_{(a_l, d_l)} \cdots R_\textsf{k}  \hookrightarrow_{(a_k, d_k)}  R_{\textsf{k+1}}\\
		p' &: R_{\textsf{0}} \hookrightarrow_{(a_0, d_0)} \cdots R_l \ \hookrightarrow_{(a_l', d_l')} \cdots R_{l+t} \hookrightarrow_{(a_{l+t}', d_{l+t}')}  R_l \ \hookrightarrow_{(a_l, d_l)} \cdots R_\textsf{k}  \hookrightarrow_{(a_k, d_k)}  R_{\textsf{k+1}}
	\end{align*}
	Here, $[\cdot]^{+}$ indicates the cycle occurring at least one time.
\end{nproperty}

Property~\ref{property} holds trivially according to Definition~\ref{def:activated_relation}.
For each undirected relation path $p$ containing a cycle,
the constraints over the tuples in $\mathcal{I}(\textbf{t}_0, p)$ are the same as the ones over the tuples in $\mathcal{I}(\textbf{t}_0, p')$.
Thus,
we can just examine the finite number of undirected relation paths in which a cycle appears at most one time.

\begin{center}
\SetAlFnt{\small}
\begin{algorithm}[t]
	\SetAlgoLined
	\SetAlgoVlined
	\SetVlineSkip{0pt}
	\SetKwIF{If}{ElseIf}{Else}{if}{:}{else if}{else}{end if}%
	\SetKwFor{While}{while}{do}{end while}%
	\SetKwFor{ForEach}{foreach}{:}{endfch}
	\SetKwFor{ForEachParallel}{foreach}{do in parallel}{endfch}
	\SetKwRepeat{Do}{do}{while}
	\SetCommentSty{redcommfont}
	\SetKwFunction{reduce}{reduce}
	\SetKwFunction{trans}{trans}
	\SetKwFunction{poll}{poll}
	
	\SetKw{Continue}{continue}
	\SetKw{Return}{return}
	\SetKw{Break}{break}
	\SetKw{Or}{or}
	\SetKw{And}{and}
	\SetKwFunction{schemaGraph}{SchemaGraph}
	\SetKwFunction{noCyclePath}{AcyclicPath}
	\SetKwFunction{cyclePath}{augmentPathWithCycle}
	\SetKwComment{Comment}{$\triangleright$\ }{}
	\newcommand\mycommfont[1]{\textcolor{violet}{#1}}
	\SetCommentSty{mycommfont}
	\SetKwProg{myproc}{Procedure}{:}{}
	\myproc{\reduce{$\Gamma$, $\mathcal{R}$, $R_p^{*}$, $R_n^{*}$}}{
		$\mathcal{R}' \leftarrow \emptyset$\; 
		$G_{\Gamma} \leftarrow \schemaGraph(\Gamma)$ \label{algline:graph}\;
		\ForEach(\label{algline:collectpath1}){$R \in \mathcal{R}$}{
				$\mathcal{P} \leftarrow \cyclePath(\noCyclePath(R^{*}, R, G_{\Gamma}))$\label{algline:collectpath2}\;
				\ForEach(\label{algline:excsecond1}){$p \in \mathcal{P}, \mathbf{t}_p \in R_p^{*}, \mathbf{t}_n \in R_n^{*}$}{
						\If{$\mathcal{I}(\mathbf{t}_p, p) \neq \mathcal{I}(\mathbf{t}_n, p)$}{
							$\mathcal{R}' \leftarrow \mathcal{R}' \cup \{ R \}$;  \Break\label{algline:excsecond2}\;
						}
				}
				\ForEach(\label{algline:excfirst1}){$p \in \mathcal{P}, \mathbf{t}_p \in R_p^{*}$}{
						\If{$\mathcal{I}(\mathbf{t}_p, p) =  \emptyset$}{
							$\mathcal{R}' \leftarrow \mathcal{R}' \setminus \{ R \}$; \Break\label{algline:excfirst2}\;
						}
				}
		}
		\Return $\mathcal{R}'$\label{algline:returnr};
	}
	\caption{Removing dummy relations for representation reduction}
	\label{alg:dummy}
\end{algorithm}
	\vspace{-9mm}
\end{center}

\smallskip
Establish upon the above concepts and property,
Algorithm~\ref{alg:dummy} shows the details of the representation reduction by removing dummy relations.
Initially, we construct the schema graph $G_{\Gamma}$ according to Definition~\ref{def:sg} (line~\ref{algline:graph}).
Then we compute the undirected relation paths from $R^{*}$ to $R$ in $G_{\Gamma}$, where any cycle repeats at most once (lines~\ref{algline:collectpath1}--\ref{algline:collectpath2}).
Specifically, the function \textsf{AcyclicPath} collects all the acyclic undirected relation paths from $R^{*}$ to $R$ in the schema graph $G_{\Gamma}$,
while the function \textsf{augmentPathWithCycle} augments each acyclic path by appending each cycle at most one time.
For each undirected relation path $p$, we compute $\mathcal{I}(\mathbf{t}_p, p)$ and $\mathcal{I}(\mathbf{t}_n, p)$ according to Definition~\ref{def:activated_relation} 
for each $\mathbf{t}_p \in R_p^{*}$ and $\mathbf{t}_n \in R_n^{*}$, respectively.
Therefore, we can identify $R$ as a non-dummy relation if both the conditions in Definition~\ref{def:dummy_relation} are violated (lines~\ref{algline:excsecond1}--\ref{algline:excfirst2}).
Finally, we obtain the reduced relational representation $\mathcal{R}'$ that excludes all the dummy relations (line~\ref{algline:returnr}).

\begin{newexample}
	Consider the undirected relation paths from \textsf{Method} to \textsf{Modifier} in Figure~\ref{fig:schema_graph}(a).
	Algorithm~\ref{alg:dummy} collects the acyclic path 
	$p_3: \textsf{Method} \hookrightarrow_{(\textsf{mdf\_id}, +1)} \textsf{Modifier}$
	and augments it to form the path $p_4: \textsf{Method} \hookrightarrow_{(\textsf{mdf\_id}, +1)} \textsf{Modifier} \hookrightarrow_{(\textsf{mdf\_id}, -1)} \textsf{Method} \hookrightarrow_{(\textsf{mdf\_id}, +1)} \textsf{Modifier}$.
Based on the activated relations $\mathcal{I}(\mathbf{t}, p_3)$ and $\mathcal{I}(\mathbf{t}, p_4)$ for each tuple $\mathbf{t}$ in \textsf{Method},
	we can find that $\mathcal{I}(\mathbf{t}_p, p) = \mathcal{I}(\mathbf{t}_n , p)$ for every $\mathbf{t}_p \in R_p^{*}$ and  $\mathbf{t}_n \in R_n^{*}$.
	Thus, \textsf{Modifier} is a dummy relation.
\end{newexample}

Essentially, our representation reduction analyzes the example program upon its relational representation.
The activated relations provide sufficient clues to identifying unnecessary relations, i.e., the dummy ones.
As the first step of the synthesis, the representation reduction
narrows down the relations used in the conjunctive query.
Furthermore, in the enumeration of the edges of a query graph, i.e., the sets $E_Q$ and $\Phi_Q$,
we only need to focus on the attributes in the non-dummy relations in the reduced relational representation,
which prunes the search space significantly.
Lastly, we formulate the soundness of the representation reduction as follows, which can further ensure the completeness of our synthesis algorithm in \Cref{subsec:algsummary}.

\begin{ntheorem}(Soundness of Representation Reduction)
	\label{thm:soundnessrr}
	If an instance of the MMCQS problem has a solution,
	there must be a conjunctive query $R_Q$, of which the Cartesian product only consists of non-dummy relations, 
	such that $R_Q$ is also a solution.
\end{ntheorem}

	\begin{proof}
	Assume that there does not exist a solution, only manipulating non-dummy relations.
	The assumption implies that there is a conjunctive query $R_Q'$ satisfying the following conditions:
	\begin{itemize}[leftmargin=*]
		\item $R_Q'$ is the solution of the MMCQS problem instance;
		\item A dummy relation $R$ appears in the Cartesian product of $R_Q'$.
	\end{itemize}
	Denote $R_Q'=\Pi_{(A_1.*)} (\sigma_{\Theta'} (\rho_{A_1}(R_1) \times \ldots \times \rho_{A_{m'}}(R_{m'})))$,
	where $\Theta'=\phi_e^{(1)} \land \cdots \land \phi_e^{(n_e')} \land \phi_s^{(1)} \land \cdots \land \phi_s^{(n_s')}$.
	Without the loss of generality, we assume that the dummy relation $R$ appears in the last few renaming expressions in the selection condition and equality/string constraints, i.e., $R_{m + 1} = R_{m + 2} = \cdots = R_{m'} = R$, and $\phi_e^{(i)}$ and  $\phi_s^{(j)}$ constrain the attributes of $R$, where $n_e < i \leq n_e'$ and $n_s < j \leq n_s'$.
	Then we can construct the query candidate $R_Q$ by removing all the occurrences of the relation $R$ from the Cartesian product and the selection condition. That is, $R_Q = \Pi_{(A_1.*)} (\sigma_{\Theta} (\rho_{A_1}(R_1) \times \ldots \times \rho_{A_{m}}(R_{m})))$,
	where $\Theta=\phi_e^{(1)} \land \cdots \land \phi_e^{(n_e)} \land \phi_s^{(1)} \land \cdots \land \phi_s^{(n_s)}$.
	As long as we prove that $R_Q$ is also a solution, 
	we can repeat the process of eliminating dummy relations iteratively until the query does not contain any dummy relations,
	which conflicts with our assumption at the beginning.
	Hence, we only need to prove that the above conjunctive query $R_Q$ is a solution, which finally proves the theorem.
	\begin{itemize}[leftmargin=*]
		\item First, the selection condition $\Theta$ in the conjunctive query $R_Q$ is weaker than the selection condition $\Theta'$,
		as $\Theta$ excludes several atomic formulas from $\Theta'$.
		\item Second, $R_1$ is renamed to $A_1$ and cannot be a dummy relation,
		which means that $R_1$ still appears in the Cartesian product of $R_Q$.
		\item Third, if the eliminated equality constraints $\phi_e^{(i)} (n_e < i \leq n_e')$ do not induce any undirected relation path from $R_1$ to $R$, the selection conditions $\Theta$ and $\Theta'$ pose the same restrictions on the tuples in $R_1$, indicating that $R_Q = R_Q'$.
		\item Fourth, if the eliminated equality constraints $\phi_e^{(i)} (n_e < i \leq n_e')$ induce a undirected relation path $p$ from $R_1$ to $R$, we need to discuss two cases as follows:
		\begin{itemize}[leftmargin=*]
			\item Case 1: $\mathcal{I}(\mathbf{t}_p, p) = \emptyset$ for a specific tuple $\mathbf{t}_p \in R_p^{*} \subset R_1 = R^{*}$. 
			According to Definition 3.4, we can obtain that $\mathbf{t}_p \notin \Pi_{(A_1.*)}\sigma_{\Theta'} (\rho_{A_1}(R_1) \times \ldots \times \rho_{A_{m'}}(R_{m'}))$
			Thus, $R_Q'$ cannot contain $\textbf{t}_p$, which conflicts with the fact that $R_Q'$ is a solution of the problem instance.
			\item Case 2: $\mathcal{I}(\mathbf{t}_p, p) = \mathcal{I}(\mathbf{t}_n, p)$ for each $\mathbf{t}_p \in R_p^{*} \subset R_1= R^{*}$ and $\mathbf{t}_n \in R_n^{*} \subset R_1= R^{*}$.
			In this case, no matter how $\phi_e^{(i)} (n_e < i \leq n_e')$ constrain the tuples in $R^{*}$, i.e., $A_1$ or $R_1$, 
			all the positive and negative tuples in $R^{*}$ belong to $\Pi_{(A_1.*)}\sigma_{\Theta'} (\rho_{A_1}(R_1) \times \ldots \times \rho_{A_{m'}}(R_{m'}))$ simultaneously.
			This also conflicts with the fact that $R_Q'$ is a solution to the problem instance.
		\end{itemize}
	\end{itemize}
	Therefore, we can always construct a conjunctive query with non-dummy relations if the problem instance has a solution.
\end{proof}


\subsection{Bounded Refinement}
\label{sec:refine}

Based on the reduced relational representation $\mathcal{R}'$,
we can enumerate query candidates by selecting proper nodes corresponding to non-dummy relations in $\mathcal{R}'$ and edges connecting such nodes.
However, the search space is potentially unbounded.
The relations can occur in a query multiple times,
i.e., a node in the schema graph can be selected more than one time.
Meanwhile, the literal in a string constraint can be instantiated flexibly,
which increases the difficulty of enumerating a query graph with proper instantiation of $\Phi_Q$.
To achieve high efficiency,
we propose the bounded refinement to expand query graphs on demand
and strengthen the query with the strongest string constraints.
Specifically, we first introduce the notions of the bounded query and the refinable query (\cref{subsec:refinable}),
and then present the details of enumerating the query graphs (\cref{subsec:enumquery}).

\subsubsection{Bounded Query and Refinable Query}
\label{subsec:refinable}

As shown in Example~\ref{ex:cq}, a relation can appear multiple times in the Cartesian product of a conjunctive query, inducing an unbounded search space in the synthesis.
The unboundedness of the search space poses the great challenge of enumerating the query candidates efficiently.
However, we realize that the conjunctive query for a code search task often involves only a few relations, each of which appears quite a few times.
Thus, it is feasible to bound the maximal multiplicity of the relation in the query and conduct the bounded enumeration.
Formally, we introduce the notion of the \emph{(m,k)-bounded query} as follows.

\begin{newdefinition}((m, k)-Bounded Query)
An $(m, k)$-bounded query is a conjunctive query with $m$ relations such that (1) each relation appears at most $k$ times; (2) there is a relation appearing exactly $k$ times.
\end{newdefinition}

\begin{newexample}
The conjunctive query $\Pi_{(A_1.*)}(\sigma_{\mathbf{true}}(\rho_{A_1}(\textsf{Method})))$ is a $(1, 1)$-bounded query.
Similarly, the conjunctive query in Example~\ref{ex:cq} is a $(4, 2)$-bounded query.
\end{newexample}

Intuitively,
we can enumerate the $(m, k)$-bounded queries by selecting non-dummy relations at most $k$ times,
forming a query graph with $m$ nodes.
When constructing the sets $E_Q$ and $\Phi_Q$ for the query graph enumeration,
we only need to concentration on the attributes of the selected relations.
However, not all the query graphs are worth enumerating.
If $R_Q$ excludes a positive tuple,
there is no need to add more nodes and edges to its query graph,
as it would induce a stronger selection condition,
making the new query still exclude the positive tuple.
Formally, we formulate the notion of the \emph{refinable query} as follows. 

\begin{newdefinition}(Refinable Query)
	A conjunctive query $R_Q$ is a refinable query if and only if for any $\mathbf{t}_p \in R_{p}^{*}$ we have $\mathbf{t}_p \in R_Q$, i.e., $R_p^{*} \subseteq R_Q$.
\end{newdefinition}

\begin{newexample}
	\label{ex:refinableQ}
	Consider $R_Q:=\Pi_{(A_1.*)} (\sigma_{\Theta} (\rho_{A_1}(\textsf{Method}) \times \rho_{A_2}(\textsf{Type}) \times \rho_{A_3}(\textsf{Parameter})))$,
	where the selection condition $\Theta$ in the selection operation is as follows:
	$$\Theta := (A_1.\textsf{id} = A_3.\textsf{method\_id}) \land (A_1.\textsf{ret\_type\_id} = A_2.\textsf{id}) \land \textsf{equal}(A_2.\textsf{name}, ``\textsf{CacheConfig}\text{''})$$	
	It is a refinable query as $R_p^{*} \subseteq R_Q = \{ (\textsf{M1}, \ \textsf{I1}, \ \textsf{T3}, \ \textsf{MDF1}),
	(\textsf{M2}, \ \textsf{I2}, \ \textsf{T3}, \ \textsf{MDF1})\}$.
\end{newexample}

Essentially, a refinable query is the over-approximation of the positive tuples.
When it excludes all the negative tuples, the query is exactly a query candidate.
Thus, we can collect the query candidates by refining the refinable queries in the bounded enumeration.

\subsubsection{Enumerating Query Candidates via Refinement}
\label{subsec:enumquery}

We denote the sets of $(m, k)$-bounded refinable queries and query candidates by $\mathcal{S}_R(m, k)$ and $\mathcal{S}_C(m, k)$, respectively,
and set a multiplicity bound $K$ to bound the multiplicity of a relation.
To conduct a bounded enumeration,
we have to compute the set $\mathcal{S}_C(m, k)$ for $k \leq K$,
which can be achieved by examining whether the queries in $\mathcal{S}_R(m, k)$ are query candidates.
Obviously, exhaustive enumeration is impossible as the search space of $(m, k)$-bounded queries is exponential to $m$ and $k$.
To avoid unnecessary enumeration,
we leverage the structure of $\mathcal{S}_C(m, k)$, which is formulated in the following property.

\begin{nproperty}
	\label{property2}
	For every refinable query $R_Q \in \mathcal{S}_C(m, k)$ and $\kappa(R_Q) := (N_Q, E_Q, \Phi_Q)$,
	there exists a query graph $G_Q^{1}$ or $G_Q^{2}$ such that
	\begin{itemize}[leftmargin=*]
		\item $N_Q = N_Q^1 \cup \{ (R, A_i) \}$, $E_Q^1 \subseteq E_Q$, 
		and $\Phi_Q^1 \subseteq \Phi_Q$,
		where $R$ appears in $G_Q^{1}$ exactly $(k-1)$ times.
		Here, $G_Q^{1} = ( N_Q^1, E_Q^1, \Phi_Q^1)$ and $\kappa^{-1}(G_Q^{1}) \in \mathcal{S}_R(m-1, k-1)$.
		\item $N_Q = N_Q^2 \cup \{ (R, A_i) \}$, $E_Q^2 \subseteq E_Q$, and $\Phi_Q^2 \subseteq \Phi_Q$,
		where $R$ appears in $G_Q^{2}$ fewer than $k$ times.
		Here, $G_Q^{2} = ( N_Q^2, E_Q^2, \Phi_Q^2)$ and $\kappa^{-1}(G_Q^{2}) \in \mathcal{S}_R(m-1, k)$.
	\end{itemize}
\end{nproperty}

Property~\ref{property2} shows that the sets of the nodes and edges in a query graph of a refinable query 
are subsumed by the ones of a refinable query with fewer relations,
which permit us to
enumerate the query candidates by computing $\mathcal{S}_R(m, k)$  and $\mathcal{S}_C(m, k)$ inductively.
Technically, we achieve the enumerative search via the bounded refinement.
Assuming that we have $\mathcal{S}_R(m', k')$ and $\mathcal{S}_C(m', k')$ for all $m' < m$ and $k' \leq k$.
Algorithm~\ref{alg:refine} computes the sets $\mathcal{S}_R(m, k)$ and $\mathcal{S}_C(m, k)$.
The technical details of the refinement are as follows:

\begin{center}
	\SetAlFnt{\small}
	\begin{algorithm}[t]
		\SetAlgoLined
		\SetAlgoVlined
		\SetVlineSkip{0pt}
		\SetKwIF{If}{ElseIf}{Else}{if}{:}{else if}{else}{end if}%
		\SetKwFor{While}{while}{do}{end while}%
		\SetKwFor{ForEach}{foreach}{:}{endfch}
		\SetKwFor{ForEachParallel}{foreach}{do in parallel}{endfch}
		\SetKwRepeat{Do}{do}{while}
		\SetCommentSty{redcommfont}
		\SetKwFunction{refine}{refine}
		\SetKwFunction{multiplicity}{multiplicity}
		\SetKwFunction{pop}{pop}
		\SetKwFunction{trans}{trans}
		\SetKwFunction{poll}{poll}
		\SetKwFunction{expand}{expand}
		\SetKwFunction{synLCS}{synLCS}
		
		\SetKw{Continue}{continue}
		\SetKw{Return}{return}
		\SetKw{Break}{break}
		\SetKw{Or}{or}
		\SetKw{And}{and}
		\SetKwFunction{schemaGraph}{getSchemaGraph}
		\SetKwFunction{noCyclePath}{AcycPath}
		\SetKwFunction{cyclePath}{augCyc}
		\SetKwComment{Comment}{$\triangleright$\ }{}
		\newcommand\mycommfont[1]{\textcolor{violet}{#1}}
		\SetCommentSty{mycommfont}
		\SetKwProg{myproc}{Procedure}{:}{}
		\myproc{\refine{$\mathcal{S}_R$, $\mathcal{S}_C$, $R_p^{*}$, $R_n^{*}$, $m$, $k$, $\mathcal{R}'$}}{
			\If{$m = 1$ \And $k = 1$}{
				$W \leftarrow \{ (\emptyset, \emptyset, \bot) \}$ \label{alg2line:graph}\;
			}
			\If(\label{alg2line1}){$m > 1$}{
				$W \leftarrow \{  \kappa(R_Q) \ | \ R_Q \in \mathcal{S}_R(m - 1, k)  \} $\;
					\If{$k > 1$}{
					$W \leftarrow W \cup \{  \kappa(R_Q) \ | \ R_Q \in \mathcal{S}_R(m - 1, k - 1)  \}$\label{alg2line2}\;
				}
			}			
			\texttt{\\}
			$\mathcal{S}_R(m, k) \leftarrow \emptyset$\; 
			\While(\label{alg2line3}){W is not empty}{
				$G_Q \leftarrow \pop(W)$; $V \leftarrow \emptyset$\;
				\ForEach(\label{alg2line4}){$R \in \mathcal{R}'$}{
					\If{$\multiplicity(G_Q, R) < k$ \And $\kappa^{-1}(G_Q) \in \mathcal{S}_R(m - 1, k)$}{
						$V \leftarrow V \cup \expand(G_Q, R)$\;
					}
					\If{$\multiplicity(G_Q, R) = (k-1)$ \And $\kappa^{-1}(G_Q) \in \mathcal{S}_R(m - 1, k-1)$}{
						$V \leftarrow V \cup \expand(G_Q, R)$\label{alg2line5} \;					
					}
				}
			    			\texttt{\\}
				\ForEach(\label{alg2line6}){$G_Q: (N_Q, E_Q, \Phi_Q) \in V$ \And $R^{*}_p \subseteq \kappa^{-1}(G_Q)$}{
					$\mathcal{S}_R(m, k) \leftarrow \mathcal{S}_R(m, k) \cup \kappa^{-1}(G_Q)$ \label{alg2line10} \;
					\ForEach(\label{alg2line7}){$(R_i, A_i) \in N_Q$ \And $(a, \textsf{STR}) \in \Gamma(R_i)$}{
						$(p, \ell) \leftarrow \synLCS(G_Q, R_i, A_i, a, R_p^{*})$\label{alg2line11}\;
						$G_Q' \leftarrow (N_Q, E_Q, \Phi_Q [(R_i, A_i, a) \mapsto (p, \ell)])$\label{alg2line12} \;
						$\mathcal{S}_R(m, k) \leftarrow \mathcal{S}_R(m, k) \cup \kappa^{-1}(G_Q')$\label{alg2line8}\;
					}
				}
			}
					\texttt{\\}
			$\mathcal{S}_C(m, k) \leftarrow \{ R_Q \ | \ R_Q \in \mathcal{S}_R(m, k), R_p^{*} = R_Q  \}$\label{alg2line9}\;
		}
		\caption{Enumerating query candidates via refinement}
		\label{alg:refine}
	\end{algorithm}
	\vspace{-8mm}
\end{center}

\begin{itemize}[leftmargin=*]
	\item For the base case, where $m=1$ and $k=1$, we construct an empty query graph (line \ref{alg2line:graph}).
	For a general case, we merge the sets of the query graphs induced by the refinable queries in $\mathcal{S}_R(m-1, k)$ and $\mathcal{S}_R(m-1, k-1)$ (lines \ref{alg2line1}--\ref{alg2line2}).
	Hence, we obtain a set of query graphs $W$ to maintain all the refinable queries with $m$ relations.
	\item For each query graph $G_Q \in W$, we leverage the function \textsf{expand} wraps a specific relation $R$ as a new node and add new edges,
	producing a set of query graphs containing the relation $R$ (lines \ref{alg2line4}--\ref{alg2line5}).
	Such query graphs are maintained in the set $V$.
	\item For each query graph $G_Q \in V$ that induces a refinable query,
	we add the induced refinable query to the set $\mathcal{S}_R(m, k)$ (line~\ref{alg2line10}) and attempt to synthesize a string constraint (lines~\ref{alg2line7}--\ref{alg2line8}).
	Specifically, the function \textsf{synLCS} examines the values of a string attribute $a$ in the positive tuples
	to compute the longest common substring $\ell$ and the strongest string predicate $p$ (line~\ref{alg2line11}),
	where all the values of $a$ in the positive tuples satisfy the string constraints induced by $p$ and $\ell$.
	By updating $\Phi_Q$, we add a new edge to $G_Q$ and produce a new query graph inducing a refinable query (lines~\ref{alg2line12}--\ref{alg2line8}).
	\item Finally, we check whether a $(m, k)$-bounded refinable query is a query candidate or not, which yields the set $\mathcal{S}_C(m, k)$ (line~\ref{alg2line9}).
	The sets $\mathcal{S}_R(m, k)$ and $\mathcal{S}_C(m, k)$ are exactly the sets of  $(m, k)$-bounded refinable queries and query candidates.
\end{itemize}

Notably, we construct the string constraints on demands to strengthen the selection condition of 
refinable queries, which is achieved by the function \textsf{synLCS} at line~\ref{alg2line11}.
The reduction from string constraint synthesis to the LCS computation enable us 
to leverage existing algorithms, such as general suffix automaton~\cite{MohriMW09},
to compute the string literal $\ell$ and select a string predicate $p$ efficiently,
which promote the efficiency of query refinement.

\begin{newexample}
	Figure~\ref{fig:refinement_example} shows the part of the refinement 
	for the instance in Example~\ref{ex:problem}.
	$G_Q^{1}$ only contains the relation \textsf{Method}.
	After adding the nodes and edges,
	we construct the query graphs of refinable queries with more relations and atomic constraints,
	such as $G_Q^{2}$, $G_Q^{3}$, and $G_Q^{4}$.
	Particularly, we identify the query candidate $\kappa^{-1}(G_Q^{4})$, i.e., $R_Q$ in Example~\ref{ex:cq}.
\end{newexample}

\begin{figure}[t]
	\centering
	\includegraphics[width=\linewidth]{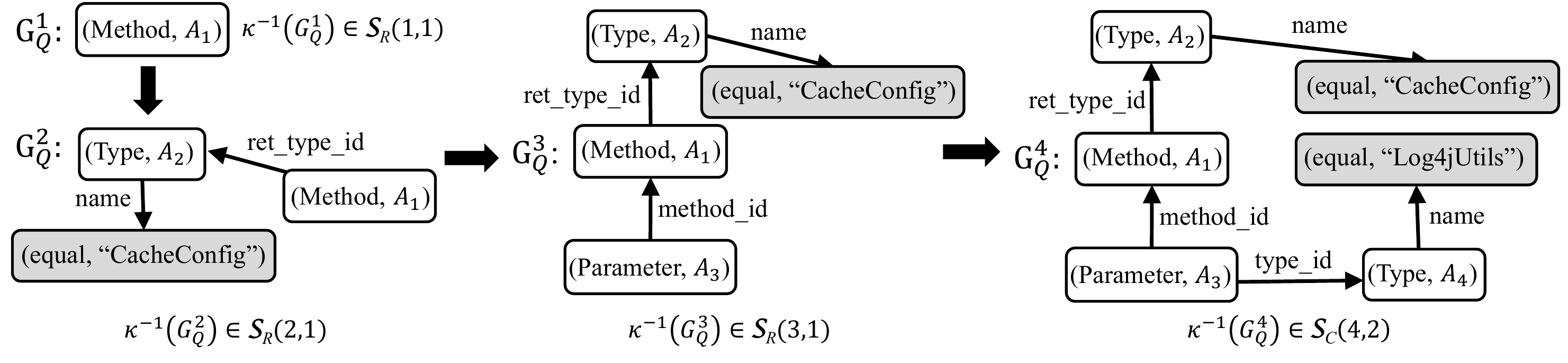}
	\vspace{-6mm}
	\caption{The example of the bounded refinement}
	\vspace{-4mm}
	\label{fig:refinement_example}
\end{figure}

\subsection{Candidate Selection}
\label{sec:prior}

Based on the bounded refinement,
we collect the query candidates in which each relation appears no more than $K$ times, 
where $K$ is the multiplicity bound.
However, not all the query candidates are the desired ones.
In this section,
we introduce dual quantitative metrics to prioritize queries (\cref{subsec:dualMetric})
and select query candidates during the refinement (\cref{subsec:blending}).

\subsubsection{Dual Quantitative Metrics}
\label{subsec:dualMetric}
Desired queries are expected to express the search intent correctly,
covering as many grammatical concepts in the natural language as possible.
Also, they should be as simple as possible according to Occam's razor.
Based on the intuitions, we formalize the following two metrics
and then define the total order to prioritize the query candidates.

\begin{newdefinition}(Named-Entity Coverage)
\label{def:nec}
Given a function $h$ mapping an attribute of a relation to a set of words in natural language,
the named entity coverage of a conjunctive query $R_Q$ with respect to a natural language description $s$ is
$$\alpha_h(R_Q, s) = \frac{1}{|N(s)|} \cdot \big| \bigcup_{(R_i, a_j) \in \mathcal{A}(\Theta)} h(R_i, a_j) \cap N(s) \big|$$
Here, $\mathcal{A}(\Theta)$ contains the relations and their attributes appearing in the selection condition $\Theta$
while $N(s)$ is the set of the named entities in the description $s$.
\end{newdefinition}

\begin{newdefinition}(Structural Complexity)
Let $\delta(\Theta)$ be the number of the atomic formulas in the selection condition $\Theta$.
The structural complexity of a conjunctive query
$R_Q$ is
$$\beta(Q) = m + \delta(\Theta)$$
\end{newdefinition}

To compute the named entity coverage,
we instantiate the function $h$ manually 
and obtain the set $N(s)$ from the natural language description $s$ based on the named entity recognition techniques~\cite{ManningSBFBM14}.
The computation does not introduce much overhead,
as the natural language description exhibits a fairly small length in our scenarios. 
Meanwhile, measuring the structural complexity is quite straightforward.
The quantities $m$ and $\delta(\Theta)$ can be totally determined by the sizes of the sets $N_Q$, $E_Q$ and $\Phi_Q$ in a query graph.
Thus, computing the two quantities does not introduce much overhead during the enumeration.

\begin{newexample}
	\label{ex:metrics}
	Assume that we instantiate the function $h$ as follows:
	$$h(\textsf{Parameter}, \textsf{id}) = \{\textsf{parameter}\}, \ h(\textsf{Method}, \textsf{id}) = \{\textsf{method}\}, \ h(\textsf{Method}, \  \textsf{idf\_id}) = \{\textsf{identifier}\}$$
	$$h(\textsf{Method}, \textsf{ret\_type\_id}) = \{\textsf{return}, \textsf{type}\}, \ h(\textsf{Method}, \textsf{mdf\_id}) = \{\textsf{modifier}\}$$
	Consider the natural language description $s$ in Example~\ref{ex:problem} and the conjunctive query $R_Q$ in Example~\ref{ex:cq},
	of which the query graph is shown in Figure~\ref{fig:schema_graph}(b).
	We can obtain a set of named entities $N(s) = \{ \textsf{method}, \textsf{type}, \textsf{parameter}, \textsf{return}\}$.
	Therefore, we  have $\alpha_h(R_Q, s) = 1$.
	According to $m = 4$ and $\delta(\Theta) = 5$, 
	its structural complexity is $\beta(R_Q) = 4 + \delta(\Theta) = 9$.
\end{newexample}

Intuitively, the selection condition of a query is more likely to conform to the user intent if the query has a higher named entity coverage.
Besides, the simpler form query can have better generalization power among the query candidates covering the same number of named entities.
Based on Occam's razor, we should choose the simplest query from the candidates that maximizes the named entity coverage.
Thus, we propose the \emph{total order} of conjunctive queries as follows.


\begin{newdefinition}(Total Order)
	\label{def:selection_criteria}
	Given the function $h$ in Definition~\ref{def:nec},
	we have $R_Q^{2} \preceq_s R_Q^{1}$
	if and only if they satisfy one of the following conditions:
	$$\alpha_h(R_Q^{1}, s) \geq \alpha_h(R_Q^{2}, s) \ \   \text{or}$$
	$$\alpha_h(R_Q^{1}, s) = \alpha_h(R_Q^{2}, s), \ \beta(R_Q^{1}) \leq \beta(R_Q^{2})$$
\end{newdefinition}

\begin{newexample}
	\label{ex:query_priority}
	Consider the following query candidates for the instance in Example~\ref{ex:problem}.
	$$R_Q^{c1} := \Pi_{(A_1.*)} (\sigma_{A_1.\textsf{name}} = ``\textsf{foo}\text{''} (\rho_{A_1}(\textsf{Method})))$$
	$$R_Q^{c2} := \Pi_{(A_1.*)} (\sigma_{\Theta_2} (\rho_{A_1}(\textsf{Method}) \times \rho_{A_2}(\textsf{Type}) \times \rho_{A_3}(\textsf{Parameter}) \times \rho_{A_4}(\textsf{Type})))$$
	Here, $\Theta_2 := \Theta \land (A_1.\textsf{name} = ``\textsf{foo}\text{''})$ and $\Theta$ is shown in Example~\ref{ex:cq}.
	Given the function $h$ shown in Example~\ref{ex:metrics},
	we can obtain that $\alpha_h(R_Q^{c1}, s) = \frac{1}{4}$, $\alpha_h(R_Q^{c2}, s) = 1$, $\beta(R_Q^{c1}) = 2$,
	and $\beta(R_Q^{c2}) = 10$.
	According to Example~\ref{ex:metrics}, 
	we have $R_Q^{c1} \preceq_s R_Q^{c2} \preceq_s R_Q$.
\end{newexample}

The total order is an adaption of Occam's razor for our synthesis problem.
Without the named entity coverage,
we would select the query candidates with the lowest structural complexity, such as $R_Q^{c1}$ in Example~\ref{ex:query_priority}, even if they do not constrain the relationship of several grammatical constructs as expected.
Based on the total order, we can select the query candidates by solving the dual-objective optimization problem, which finally yields the desired queries for code search.

\subsubsection{Blending Selection with Refinement}
\label{subsec:blending}

Based on Definition~\ref{def:selection_criteria},
we propose Algorithm~\ref{alg:select} that blends the candidate selection with the bounded refinement,
which is more likely to obtain the desired queries for a code search task.
First, we obtain the schema graph and remove the dummy relations via the representation reduction (line~\ref{algline:reduction2}).
We then compute the upper bound of the named entity coverage, which is denoted by $\widetilde{\alpha}$ (line~\ref{algline:upperbound}).
After the initialization of $\alpha_{max}$, $\beta_{min}$, and $\mathcal{S}_Q$ (line~\ref{algline:init_start}),
we conduct the bounded refinement and select the query candidates in each round (lines~\ref{algline:bound1}--\ref{algline:terminate2}).
Obviously, there are at most $K \cdot |\mathcal{R}'|$ relations in a query for a given multiplicity bound $K$ (line~\ref{algline:bound1}),
and a relation can only appear at most $min(K, m)$ times in a query with $m$ relations (line~\ref{algline:queryselection8}).
In each round, we fuse the refinement and selection as follows:

\begin{center}
	\SetAlFnt{\small}
	\begin{algorithm}[t]
		\SetAlgoLined
		\SetAlgoVlined
		\SetVlineSkip{0pt}
		\SetKwIF{If}{ElseIf}{Else}{if}{:}{else if}{else}{end if}%
		\SetKwFor{While}{while}{do}{end while}%
		\SetKwFor{ForEach}{foreach}{:}{endfch}
		\SetKwFor{ForEachParallel}{foreach}{do in parallel}{endfch}
		\SetKwRepeat{Do}{do}{while}
		\SetCommentSty{redcommfont}
		\SetKwFunction{select}{synthesize}
		\SetKwFunction{trans}{trans}
		\SetKwFunction{poll}{poll}
		
		\SetKw{Continue}{continue}
		\SetKw{Return}{return}
		\SetKw{Break}{break}
		\SetKw{Or}{or}
		\SetKw{And}{and}
		\SetKwFunction{schemaGraph}{getSchemaGraph}
		\SetKwFunction{noCyclePath}{AcycPath}
		\SetKwFunction{cyclePath}{augCycle}
		\SetKwFunction{dummy}{reduce}
		\SetKwFunction{update}{update}
		\SetKwFunction{refine}{refine}
		\SetKwComment{Comment}{$\triangleright$\ }{}
		\newcommand\mycommfont[1]{\small\textcolor{violet}{#1}}
		\SetCommentSty{mycommfont}
		\SetKwProg{myproc}{Procedure}{:}{}
		\myproc{\select{$\Gamma$, $\mathcal{R}$, $R_p^{*}$, $R_n^{*}, s, K$}}{
			$\mathcal{R}' \leftarrow \dummy(\Gamma, \mathcal{R}, R_p^{*}, R_n^{*})$ \label{algline:reduction2}\;
			$\widetilde{\alpha} \leftarrow \frac{1}{|N(s)|}|  \{w \in h(R, a) \cap N(s) \ | \ \exists R \in \mathcal{R}', T \in \mathcal{R}' \cup \{\textsf{STR}\}: (a, T) \in \Gamma(R) \}|$\label{algline:upperbound}\;
			$\alpha_{max} \leftarrow \textsf{MIN\_INT}$\label{algline:init_start}; 
			$\beta_{min} \leftarrow \textsf{MAX\_INT}$; 
			$\mathcal{S}_{Q} \leftarrow \emptyset$\label{algline:init_end}\;
			\ForEach(\label{algline:bound1}){$1 \leq m \leq K \cdot |\mathcal{R}'|$}{
				\If(\label{algline:terminate5}){$\mathcal{S}_R(m-1, k-1) = \emptyset$ \And $ \mathcal{S}_R(m-1, k) = \emptyset$}{
					\Continue\label{algline:terminat6}\;
				}
				\ForEach(\label{algline:queryselection8}){$1 \leq k \leq min(K, m)$}{
					$\refine(\mathcal{S}_R, \mathcal{S}_C, R_p^{*}, R_n^{*}, m, k, \mathcal{R}')$\label{algline:refine}\;
					\ForEach(\label{algline:queryselection1}){$R_Q \in \mathcal{S}_{C}(m, k)$}{
						($\alpha_{max}$, $\beta_{min}$, $\mathcal{S}_Q$) $\leftarrow$ \update($R_Q$, $\alpha_{max}$, $\beta_{min}$, $\mathcal{S}_Q$)\label{algline:queryselection2}\;
						}
					$\widetilde{\beta} =  \min (\{ \beta(R_Q) \ | \ R_Q \in \mathcal{S}_R(m, k) \cup \mathcal{S}_R(m, k-1)   \})$\;
					\If(\label{algline:terminate1}){$\alpha_{max} = \widetilde{\alpha}$ \And $\beta_{min} \leq \widetilde{\beta}$}{
						\Return{$\mathcal{S}_Q$}\label{algline:terminate2}\;
					}
				}
				}
				\Return{$\mathcal{S}_Q$}\;
		}
		\caption{Blending selection with refinement}
		\label{alg:select}
	\end{algorithm}
\vspace{-9mm}
\end{center}


\begin{itemize}[leftmargin=*]	
	\item Enumerate $(m, k)$-bounded refinable queries and query candidates with Algorithm~\ref{alg:refine}, strengthening the selection conditions of the refinable queries in previous rounds (line~\ref{algline:refine}).
	\item Compute $\alpha_h(R_Q, s)$ and $\beta(R_Q)$ for each $(m, k)$-bounded query candidate $R_Q$
	and update the selected candidate set $\mathcal{S}_Q$, $\alpha_{max}$, and $\beta_{min}$ (lines~\ref{algline:queryselection1}--\ref{algline:queryselection2}).
	Particularly, $\alpha_{max}$ and $\beta_{min}$ are updated to identify the largest candidates with respect to the total order.
	\item Terminate the iteration in advance and return the set $\mathcal{S}_Q$ if $\alpha_{max}$ reaches the upper bound of the named entity coverage, i.e., $\widetilde{\alpha}$,
	and the queries to be refined in the next round do not have lower structural complexities than $\beta_{min}$
	(lines~\ref{algline:terminate1}--\ref{algline:terminate2}).
\end{itemize}

\vspace{-3mm}

The refinement strengthens the selection conditions
to exclude all the negative tuples.
Specifically,
we explore the bounded search space containing the query graphs of refinable queries,
avoiding the unnecessary enumerative search effectively.
In real-world code search tasks, the selection condition is often involved different kinds of grammatical constructs,
making each relation appear often appear in the conjunction query one or two times.
Therefore, we set the multiplicity bound $K$ to 2 for real-world code search tasks in practice,
of which the effectiveness will be evidenced by our evaluation.

The natural language description benefits our synthesis process from two aspects.
First, the selected queries in $\mathcal{S}_Q$ are the largest queries under the total order,
and thus they are more likely to conform to the user's search intent than other query candidates.
Second, we terminate the enumerative search if the named entity coverage cannot increase with a smaller structural complexity,
avoiding unnecessary enumerative search of bounded query candidates for the efficiency improvement.

\begin{newexample}
Consider $R_Q^{c1}$, $R_Q^{c2}$ and $R_Q$ in Example~\ref{ex:query_priority}.
We obtain the query candidate $R_Q^{c1}$ when $(m, k) = (1, 1)$, and discover the candidates $R_{Q}$ and $R_Q^{c2}$ when $(m, k) = (4, 2)$.
Based on Definition~\ref{def:selection_criteria}, we select and maintain the query candidate $R_Q$ in $\mathcal{S}_Q$.
Also, we find $\alpha_{max} = \alpha_h(R_Q, s)=1$ reaches $\widetilde{\alpha}$,
indicating that the candidates in the subsequent rounds cannot yield a larger named entity coverage
with lower structural complexity.
Algorithm~\ref{alg:select} terminates and returns the set $\mathcal{S}_Q = \{ R_Q \}$.
\end{newexample}

\subsection{Summary}
\label{subsec:algsummary}

Our synthesis algorithm \ToolName\ is an instantiation of a new synthesis paradigm of the multi-modal synthesis, 
which reduces the synthesis problem to a multi-target optimization problem.
We now formulate and prove the soundness, completeness, and optimality of \ToolName\ with 
three theorems as follows.

\begin{ntheorem}(Soundness)
	\label{thm:soundness}
	For any $R_Q \in \mathcal{S}_Q$, where $\mathcal{S}_Q$ is returned by Algorithm~\ref{alg:select},
	$R_Q$ must contain all the positive tuples in $R_p^{*}$
	and exclude the negative tuples in $R_n^{*}$.
\end{ntheorem}

\begin{proof}
	To prove the soundness of \ToolName,
	we only need to prove that 
	for any $R_Q \in \mathcal{S}_C(m, k)$,
	$R_Q$ must contain all the positive tuples in $R_p^{*}$
	and exclude all the negative tuples in $R_n^{*}$.
	According to line \ref{alg2line9} in Algorithm~\ref{alg:refine},
	we only add a conjunctive query $R_Q$ to the set $\mathcal{S}_C(m, k)$ 
	if $R_Q$ only contains the positive tuples in $R_p^{*}$.
	Therefore, $R_Q$ is a query candidate.
	The soundness of our algorithm is proved.
\end{proof}

\begin{ntheorem}(Completeness)
	\label{thm:completeness}
	If an MMCQS problem instance has an $(m, k)$-bounded query as its solution and $k \leq K$,
	the set $\mathcal{S}_Q$ returned by Algorithm~\ref{alg:select} is not empty.
\end{ntheorem}

\begin{proof}
	Assume that there exists an $(m, k)$-bounded query $R_Q$ as the solution of the MMCQS problem instance.
	According to Theorem~\ref{thm:soundnessrr}, 
	we can construct another query candidate $R_Q'$ such that
	\begin{itemize}[leftmargin=*]
		\item No dummy relation appears in the Cartesian product of $R_Q'$.
		\item The query graph of $R_Q'$ is the subgraph of the query graph of $R_Q$.
	\end{itemize}
	Meanwhile, Algorithm~\ref{alg:select} invokes Algorithm~\ref{alg:refine} inductively to compute all the $(m, k)$-bounded queries
	and query candidates, where $k \leq K$.
	Also, we notice that $R_Q'$ must belong to $\mathcal{S}_R(m, k)$ and $\mathcal{S}_C(m, k)$ for some $(m, k)$,
	which implies that $\mathcal{S}_C$ can not be empty.
	Hence, the completeness of \ToolName\ is proved.
\end{proof}

\begin{ntheorem}(Optimality)
	\label{thm:optimality}
	Denote $I = \{(m, k) \ | \ 1 \leq m \leq K \cdot |\mathcal{R}'|, 1 \leq k \leq min(K, m) \}$
	and $\widetilde{\mathcal{S}} = \bigcup_{(m, k)\in I} \mathcal{S}_C(m, k)$.
	The returned query set $\mathcal{S}_Q$ of Algorithm~\ref{alg:select} satisfies:
	\begin{itemize}[leftmargin=*]
		\item $R_Q' \preceq_s R_Q$ for every $R_Q \in \mathcal{S}_Q$ and $R_Q' \in \widetilde{\mathcal{S}}$.
		\item There do not exist $R_Q \in \widetilde{\mathcal{S}} \setminus \mathcal{S}_Q$ and $R_Q' \in \widetilde{\mathcal{S}}$ such that $R_Q' \preceq_s R_Q$ and $R_Q \npreceq_s R_Q'$.
	\end{itemize}
\end{ntheorem}

	\begin{proof}
	We first introduce a set $I'$ to contain all the pairs $(m, k)$ iterated in the executed rounds of Algorithm~\ref{alg:select}.
	Denote $\widetilde{\mathcal{S}'} = \bigcup_{(m, k)\in I'} \mathcal{S}_C(m, k)$.
	According to the functionality of the method \textsf{update} invoked at line~\ref{algline:queryselection2} in Algorithm~\ref{alg:select},
	it has selected all the largest query candidates based on the total order in Definition~\ref{def:selection_criteria}.
	Hence, we can obtain that 
	\begin{itemize}[leftmargin=*]
		\item (P1) $R_Q' \preceq_s R_Q$ for any $R_Q \in \mathcal{S}_Q$ and $R_Q' \in \widetilde{\mathcal{S}'}$
		\item (P2) There do not exist $R_Q \in \widetilde{\mathcal{S}^{'}} \setminus \mathcal{S}_Q$ and $R_Q' \in \widetilde{\mathcal{S}^{'}}$ such that $R_Q' \preceq_s R_Q$ and $R_Q \npreceq_s R_Q'$.
	\end{itemize}
	If all the pairs $(m, k) \in I$ are iterated, we have $\widetilde{\mathcal{S}} = \widetilde{\mathcal{S}'}$.
	The theorem holds trivially.
	If Algorithm~\ref{alg:select} terminates from line~\ref{algline:terminate2}, several pairs in $I$ are not iterated,
	leaving several query candidates not enumerated.
	In what follows, we will prove the two properties in the theorem one by one for this case.
	
	First, we try to prove $R_Q' \preceq_s R_Q$ for any $R_Q \in \mathcal{S}_Q$ and $R_Q' \in \widetilde{\mathcal{S}}$,
	which comes to two cases.
	If $R'_Q \in \widetilde{\mathcal{S}'}$, we can obtain that $R_Q' \preceq_s R_Q$ according to (P1) above.
	Otherwise, there exists $(m, k) \in I \setminus I'$ such that $R'_Q \in \mathcal{S}_C(m, k)$.
	According to lines 13 and 14 in Algorithm 3,
	the query graph of $R'_Q$ must have more relations or edges than the query graphs of specific queries 
	in $\mathcal{S}_C(m-1, k-1)$ and $\mathcal{S}_C(m-1, k)$,
	which implies 
	$$\beta(R'_Q) > \min (\{   \beta(R_Q)    \ | \ R_Q \in \mathcal{S}_C(m-1, k-1) \cup \mathcal{S}_C(m-1, k)\})$$
	More generally, we have
	$$\min (\{   \beta(R_Q)    \ | \ R_Q \in \mathcal{S}_C(m, k) \}) > \min (\{   \beta(R_Q)    \ | \ R_Q \in \mathcal{S}_C(m-1, k-1) \cup \mathcal{S}_C(m-1, k)\})$$
	Therefore, we have $\beta(R'_Q) > \beta_{min}^{*}$, where $\beta_{min}^{*}$ is the value of $\beta_{min}$ when Algorithm~\ref{alg:select} terminates.
	Recap that $\beta_{min}^{*}$ is the structural complexity of the queries in $\mathcal{S}_C$,
	while the queries in $\mathcal{S}_C$ have reached $\widetilde{\alpha}$, i.e., the upper bound of the named entity coverage.
	Thus, we can easily obtain $R_Q' \preceq_s R_Q$ for any $R_Q \in \mathcal{S}_Q$ based on the total order in Definition~\ref{def:selection_criteria}.
	
	Second, we try to prove that there does not exist $R_Q \in \widetilde{\mathcal{S}} \setminus \mathcal{S}_Q$ and $R_Q' \in \widetilde{\mathcal{S}}$ such that $R_Q' \preceq_s R_Q$ and $R_Q \npreceq_s R_Q'$.
	If it does not hold, we can find $(m, k) \in I \setminus I'$ such that $R_Q \in \mathcal{S}_C(m, k)$
	and $\beta(R_Q) < \beta_{min}^{*}$.
	However, $\beta(R_Q) >  \widetilde{\beta}^{*} \geq \beta_{min}^{*}$,
	where $\widetilde{\beta}^{*}$ is the value of $\widetilde{\beta}$ when Algorithm~\ref{alg:select} terminates.
	Contradiction!
	
	Lastly, we have proved the optimality of our synthesis algorithm.
\end{proof}

\section{Implementation}
\label{sec:impl}
Established upon the industrial Datalog-based Java program analyzer in Ant Group,
\ToolName\ synthesizes conjunctive queries to support code search tasks in Java programs.
Noting that our approach is general enough to support the conjunctive query synthesis for any Datalog-based analyzer
as long as the generated relations can be formulated by Definition~\ref{def:pf}. 
In what follows,
we provide more implementation details of \ToolName.

\smallskip
\emph{\textbf{Synthesis Input Configuration.}}
We design a user interface to convenience the users to specify examples in a code snippet.
Specifically, the users can copy a desired grammatical construct from their workspace as a positive example
or write a positive example manually.
By mutating a positive example, the users can create more positive and negative examples,
eventually forming an example program.
Then we convert the program to the relational representation, which consists of 173 relations with 1,093 attributes in total, and partitions a relation into two parts to induce positive and negative tuples.
To extract the named entities from the natural language description,
we leverage the named entity recognition~\cite{ManningSBFBM14}
and construct the dictionary of entities to filter unnecessary named entities in the post-processing.
Specifically, the dictionary contains 205 words,
which are the keywords describing grammatical constructs in Java programs,
such as \text{``}method\text{''}, \text{``}parameter\text{''}, and \text{``}return\text{''}.
Furthermore, we also  instantiate the function $h$ in Definition~\ref{def:nec} to bridge the program relational representation with natural language words.
We publish all the synthesis specifications and dictionary of entities online~\cite{data}.

\smallskip
\emph{\textbf{Synthesis Algorithm Design.}}
Based on the language schema of Java, we construct the schema graph offline and persist it for synthesizing queries for a given synthesis specification.
Instead of invoking the Datalog-based analyzer,
we implement a query evaluator for conjunctive queries upon the relational representation to identify the refinable queries and query candidates,
which can improve the efficiency of the query evaluation during the synthesis.
In the bounded refinement, we set the multiplicity bound $K$ to 2 by default to support code search tasks.
To efficiently synthesize string constraints,
we leverage the generalized suffix automaton~\cite{MohriMW09} 
to identify the longest common substrings of a set of string values,
which returns the string predicate $p$ and the string literal $\ell$ with low time overhead.
Currently, \ToolName\ utilizes four predicates for string match,
while we can further extend it to support regex match by adopting existing regex synthesis techniques~\cite{LeeSO16, Chen0YDD20} 
to Algorithm~\ref{alg:refine}.


\section{Evaluation}
\label{sec:eval}

To quantify the effectiveness and efficiency of \ToolName, we conduct a comprehensive empirical evaluation
and answer the following four research questions:

\begin{itemize}[leftmargin=*]
	\item \textbf{RQ1:} How effective is \ToolName\ in the conjunctive query synthesis for code search tasks?
	\item \textbf{RQ2:} How big are the benefits of the representation reduction and the bounded refinement in terms of efficiency?
	\item  \textbf{RQ3:} Is the query candidate selection effective and necessary for the synthesis?
	\item \textbf{RQ4:} How does \ToolName\ compare to other approaches that could be used in our problem?
\end{itemize}

\textit{\textbf{Benchmark.}}
There are no existing studies targeting our multi-modal synthesis problem,
so we construct a new benchmark for evaluation,
which consists of \CaseNumber\ code search tasks.
As shown in Table~\ref{table:cases},
the tasks cover five kinds of grammatical constructs,
namely variables, expressions, statements, methods, and classes.
Specifically, 14 tasks are the variants of C++ search tasks in~\cite{NaikMSWNR21} or the query synthesis tasks in~\cite{ThakkarNSANR21}.
We also consider more advanced tasks deriving from real demands.
For example, Task 2 originates from the coding standard of a technical unit in Ant Group, while Task 21 is often conducted when the developers check the usage of the \textsf{log4j} library to improve reliability.
For each task, we specify examples in a program and a sentence as the natural language description.
The program is fed to the commercial Datalog-based analyzer in Ant Group to 
generate the relational representation and the relation partition,
while the natural language description is processed via the named entity recognition technique~\cite{ManningSBFBM14}.
The columns \textbf{L} and \textbf{(P, N)} in Table~\ref{table:cases} indicate the line numbers of the programs and the numbers of positive/negative tuples, respectively.

\smallskip
\textit{\textbf{Experimental Setup.}}
We conduct all the experiments on a Macbook Pro with a 2.6 GHz Intel\textregistered\ Core\texttrademark\ i7-9750H CPU and 16 GB physical memory.


\subsection{Overall Effectiveness}

\begin{table}[t]
	\centering
\caption{Experiment results of synthesizing conjunctive queries for code search tasks.}
	\resizebox{\linewidth}{!} 
	{
		\begin{tabular}{c|l|cc|ccc|c}
			\toprule
			\textbf{ID}  & \textbf{Description} & \textbf{L} & \textbf{(P, N)} & $\mathbf{|G_Q|}$  & \textbf{k} &  $|\mathbf{G}'_{\mathbf{\Gamma}}|$& $\mathbf{T_0}$\textbf{(s)} \\ \toprule
			1 & Local variables with double type & 10 & (2, 2) & (2, 1, 1)     &          1   & (10, 19)                  & 2.36          \\
			2 &  Float variables of which the identifier contains ``cash\text{''}&10&  (3, 1)  & (3, 2, 2)   &        1     & (9, 17)             & 2.37      \\
			3 & Public field variables of a class &6& (2, 1)  & (2, 1, 1)   &          1   & (10, 21)            & 2.30       \\
			4 & Public field variables whose names use ``cash\text{''} as suffixes&8&(3, 2) & (2, 2, 2)   &          1      & (10, 23)               & 2.49   \\
			\midrule
			5 & Arithmetic expressions using double-type operands&9 &(2, 2) & (3, 2, 2)  &          1      & (9, 29)                & 2.62   \\
			6 & Cast expressions from double-type to float-type~\cite{ThakkarNSANR21} &11&(1, 2) & (3, 2, 2)   &         2       & (9, 23)             & 2.31   \\
			7 & Arithmetic expressions only using literals as operands&19 &(2, 3) & (3, 2, 1)  &            2    & (9, 29)           & 2.76 \\
			8 & Expressions comparing a variable and  a boolean literal~\cite{NaikMSWNR21} &16&(2, 1) & (2, 1, 1)  &       1       & (8, 29)            & 2.36  \\
			9 & New expressions of ArrayList &8&(2, 1) & (2, 1, 1)  &         1       & (10, 25)          & 2.20 \\
			10 & Logical conjunctions with a boolean literal~\cite{NaikMSWNR21} &11&(3, 1) & (2, 1, 2)  &           1     & (9, 29)            & 2.31  \\
			11 & Float increment expression~\cite{ThakkarNSANR21} &11&(1, 2) & (3, 2, 2)  &           1     & (9, 23)            & 2.55\\
			12 & Expressions comparing two strings with ``==\text{''}~\cite{NaikMSWNR21} & 14 & (2, 1) & (3, 2, 3) &1 & (11, 46) & 3.01\\
			13 & Expressions performing downcasting~\cite{ThakkarNSANR21}  & 25 & (2, 1) & (3, 2, 0) & 1 & (11, 39) &2.63 \\
			\midrule
			14 & The import of LocalTime & 7&(1, 1)& (1, 0, 2)  &          1       & (9, 23)          & 2.17   \\
			15 & The import of the classes in log4j &9&(3, 1) & (1, 0, 1)  &        1     & (9, 22)          & 2.22 \\
			16 & Labeled statements using ``err\text{''} as the label~\cite{NaikMSWNR21} & 17 &(1, 1)& (1, 0, 1)  &          1     & (10, 21)       & 2.18   \\
			17 & If-statements with a boolean literal as a condition~\cite{NaikMSWNR21} &16&(2, 1)& (2, 1, 0)   &         1      & (9, 17)        & 2.24 \\
			18 & For-statements with a boolean literal as the condition~\cite{NaikMSWNR21} &15&(2, 1)& (2, 1, 0)   &          1    & (10, 25)           & 2.31 \\
			19 & Invocation of unsafe time function ``localtime\text{''}~\cite{NaikMSWNR21} & 9 & (2, 1)& (2, 1, 1) & 1 & (10, 23) & 2.23\\
			\midrule
			20 & Public methods with void return type~\cite{NaikMSWNR21} &10&(2, 1)& (3, 2, 2)    &           1    & (11, 26)          & 2.36 \\
			21  & Methods receiving a parameter with Log4jUtils type &11&(2, 1)& (3, 2, 1)     &        1     & (9, 20)         & 2.45 \\
			22 & Methods using a boolean parameter as a if-condition~\cite{ThakkarNSANR21} &29& (2, 2) & (4, 3, 0)  &        1      & (11, 26)        & 3.23\\
			23 & Methods creating a File object & 14 &(2, 1) & (3, 2, 1)    &      1     & (12, 30)             & 2.38  \\
			24 & Mutually recursive methods~\cite{NaikMSWNR21, ThakkarNSANR21} & 20 & (2, 2) & (2, 2, 0)  &           2   & (11, 27)            & 2.42\\
			25 & Overriding methods of classes~\cite{ThakkarNSANR21} & 25 & (2, 4)  &  (5, 5, 0)  &  2  &  (8, 22)  & 5.89 \\
			\midrule
			26 & User classes with ``login\text{''} methods &15&(2, 1) & (2, 1, 2)   &          1    & (11, 28)            & 2.53\\
			27 & Classes containing a field with Log4jUtils type &20&(2, 1)& (3, 2, 1)   &         1    & (12, 35)           & 2.42  \\
			28 & Classes having a subclass &25& (3, 3) & (2, 1, 0)  &       2    & (13, 33)        & 2.21    \\
			29 & Classes implementing \textsf{Comparable} interface & 16 &(2, 1) & (2, 1, 1)   &         1     & (13, 35)            & 2.58  \\
			30 & Classes containing a static method & 17 &(2, 1) & (3, 2, 1)  &       1     & (11, 28)             & 2.46  \\
			31 & Java classes with main functions & 16& (2, 1) & (2, 1, 1) & 1 & (10, 28) & 2.35\\
			\bottomrule
		\end{tabular}
	}
	\label{table:cases}
\end{table}

To evaluate the effectiveness of \ToolName,
we run it upon the synthesis specification for each code search task,
examining whether the synthesized queries express the intent correctly,
and meanwhile,
measure the time cost of synthesizing queries in each task.

In Table~\ref{table:cases}, 
the column $|\mathbf{G_Q}|$ indicates the numbers of the relations, equality constraints, and string constraints.
The column \textbf{k} shows the maximal multiplicity of a relation in a synthesized query,
while the column $|\mathbf{G}'_{\mathbf{\Gamma}}|$ indicates the numbers of nodes and edges in
the subgraph of the schema graph induced by $\mathcal{R}' \cup \{ \textsf{STR}\}$.
The time cost of \ToolName\ is shown in the column $\textbf{T}_0$.
According to the statistics, we can obtain two main findings.
First,  \ToolName\ synthesizes the queries for all the tasks successfully. 
It manipulates more than three relations in Tasks 22 and 25, which are even non-trivial for a human to achieve. 
Second, \ToolName\ synthesizes the queries with a quite low time cost.
The average time cost is 2.56 seconds, 
while most of the tasks are finished in three seconds.

As mentioned in~\cref{sec:impl},
\ToolName\ performs the bounded refinement with the multiplicity bound $K = 2$.
In our benchmark, five code search tasks demand several relations appear two times.
In practice, the searching condition can hardly relate to more than two grammatical constructs of the same kind,
so our setting of the multiplicity bound $K$ enables \ToolName\ to synthesize queries for code search tasks in real-world scenarios.
Meanwhile, we quantify the time cost of the synthesis in the cases of $K=3$ and $K=4$.
Averagely, \ToolName\ takes 2.71 seconds and 3.98 seconds under the two settings, respectively.
Thus, the overhead increases gracefully when $K$ increases,
demonstrating the great potential of \ToolName\ in efficiently synthesizing more sophisticated queries with a larger multiplicity bound.

\subsection{Ablation Study on Efficiency}

We evaluate two ablations of \ToolName,
namely \ToolNameNRR\ and \ToolNameNQR,
to quantify the impact of the representation reduction and the bounded refinement on the efficiency.
\begin{itemize}[leftmargin=*]
\item \ToolNameNRR: This ablation of \ToolName\ does not perform the representation reduction but still leverages Algorithm~\ref{alg:refine} to conduct the bounded refinement.

\item \ToolNameNQR: The ablation performs the representation reduction as \ToolName\ does, while it enumerates all the query graphs and permits each relation to appear at most $K$ times.
\end{itemize}
We measure the time cost of two ablations to quantify their efficiency.
Specifically,
we set the time budget for synthesizing queries for a single task to 30 seconds,
as a synthesizer would have little practical value for the real-world code search
if it ran out of the time budget.

\begin{figure}[t]
	\centering
	\includegraphics[width=0.75\linewidth]{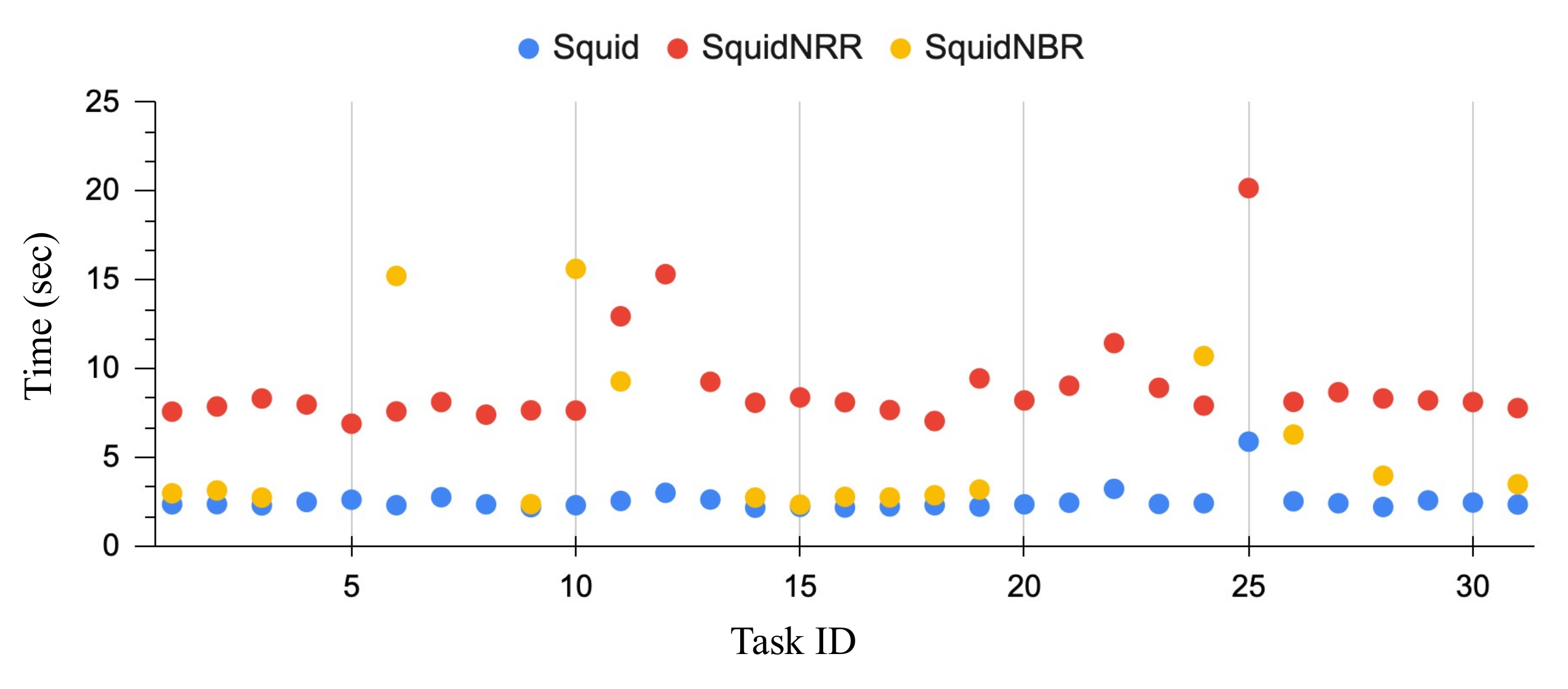}
	\vspace{-3mm}
	\caption{The time cost comparison of \ToolName, \ToolNameNRR, and \ToolNameNQR}
	\label{fig:RQ2}
\end{figure}

Figure~\ref{fig:RQ2} shows the comparison of the time cost of \ToolName\ and the two ablations.
First, the representation reduction can effectively reduce the time cost.
Specifically, the average time cost of \ToolNameNRR\ is 8.98 seconds,
indicating that the representation reduction introduces a 71.49\% reduction over the time cost.
Second, the bounded refinement has a critical impact on the efficiency of \ToolName.
Without the refinement, \ToolNameNQR\ has to explore the huge search space
induced by the non-dummy relations,
making 14 out of 31 tasks cannot be finished within the time budget,
such as Task 4, Task 5, etc.
For the failed tasks, we do not show the time cost of \ToolNameNQR\ in Figure~\ref{fig:RQ2}.
\ToolNameNQR\ also takes much more time than \ToolName,
consuming 7.89 seconds on average,
even if it successfully synthesizes the queries.

To investigate how the efficiency is improved,
we further measure the size of the subgraph of the schema graph induced by $\mathcal{R}' \cup \{ \textsf{STR}\}$.
Initially, the schema graph contains 174 nodes (including the node depicting \textsf{STR}) and 1,093 edges.
As shown in the column $|\mathbf{G}'_{\mathbf{\Gamma}}|$ of Table~\ref{table:cases},
the induced subgraph only contains around ten nodes and no more than fifty edges.
Although \ToolNameNRR\ prunes unnecessary relations by enumerating several bounded queries at the beginning of the refinement,
it has to spend more time on the query enumeration than \ToolName,
which demonstrates the critical role of the bounded refinement in our synthesis.
Besides,
the running time of \ToolNameNQR\ is similar to \ToolName\ on several benchmarks, such as Task 1, Task 2, Task 3, Task 9, etc., while it takes much longer time than \ToolName\ in other benchmarks. 
Although \ToolNameNQR\ does not discard infeasible queries, it still benefits from representation reduction. When a desired query is of small size and the reduced program representation induces a small schema graph $G_{\Gamma}'$,
\ToolNameNQR\ can terminate to find an optimal query by enumerating a few candidates. 
However, if $G_Q$ and $G_{\Gamma}'$ are large, 
\ToolNameNQR\ enumerates a large number of infeasible queries, which introduces significant overhead.

\begin{figure}[t]
	\centering
	\includegraphics[width=\linewidth]{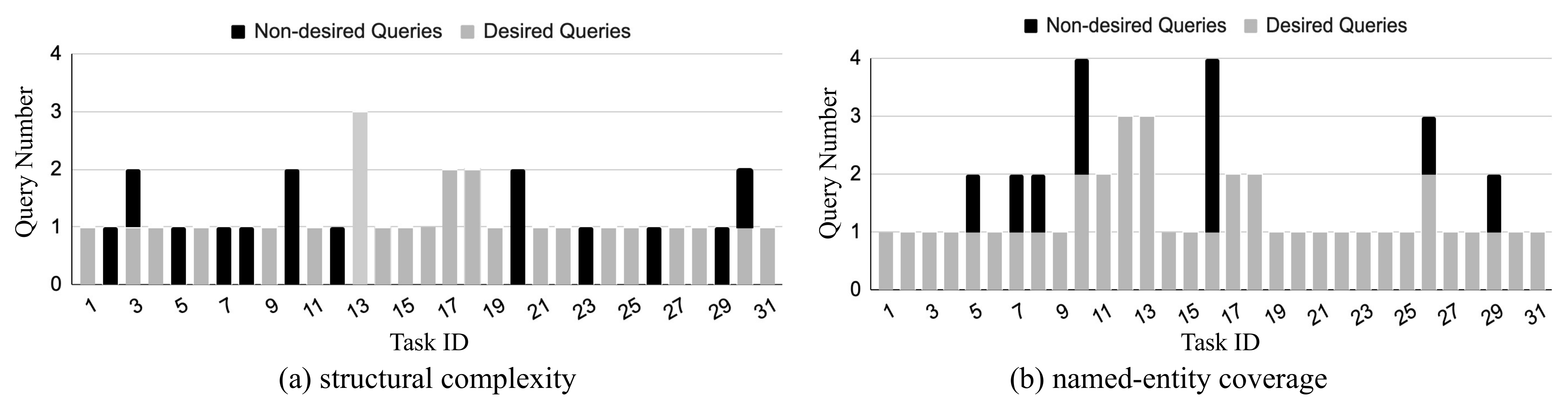}
	\vspace{-7mm}
	\caption{The numbers of synthesized queries prioritized with different metrics}   
	\vspace{-5mm}
	\label{fig:RQ3}
\end{figure}

\subsection{Impact of Selection}
To measure the impact of the query candidate selection,
we adapt each metric separately for candidate prioritization.
Specifically, we alter Algorithm~\ref{alg:select} and select the queries minimizing the structural complexity and maximizing the named entity coverage, respectively.
We then count the returned queries and inspect whether they are desired ones or not.

Figure~\ref{fig:RQ3}(a) shows the numbers of the synthesized queries with the structural complexity
as the metric.
As we can see,
\ToolName\ produces non-desired queries in 12 tasks,
while the returned set of synthesized queries in 10 tasks do not contain any desired query.
For Task 3 and Task 30, it provides the desired queries along with non-desired ones,
which makes the users confused about how to select a proper one.
Similar to $R_Q^{c1}$ in Example~\ref{ex:query_priority},
the non-desired queries are caused by the over-fitting of positive and negative examples.
Although they have the simplest form of the selection conditions,
the relationship of grammatical constructs mentioned in the natural language description is not constrained,
making synthesized queries not express the users' intent correctly.

Figure~\ref{fig:RQ3}(b) shows the numbers of the queries maximizing the named entity coverage.
\ToolName\ returns at least one non-desired query for seven tasks.
Similar to $R_Q^{c2}$ in Example~\ref{ex:query_priority},
non-desired queries come from over-complicated selection conditions.
Although the selected queries have the same named entity coverage,
several queries contain more atomic formulas than the desired ones,
posing stronger restrictions upon the code.
Besides, the synthesized queries in several tasks, e.g., Task 11 and Task 26,
have complex selection conditions
although they are equivalent under the context of code search.
However, such queries exhibit higher structural complexity,
posing more difficulty in understanding them.


\subsection{Comparison with Existing Techniques}

To the best of our knowledge, no existing technique or implemented tool targets the same problem as \ToolName.
To compare \ToolName\ with existing effort,
we adapt the state-of-the-art Datalog synthesizer~\textsc{EGS}~\cite{ThakkarNSANR21} as our baseline.
Originally, it synthesizes a conjunctive query to separate a positive tuple from all the negative ones
and then group all the conjunctive queries as the final synthesis result,
which can be theoretically a disjunctive query.
However, \textsc{EGS} does not synthesize string constraints and only prioritizes feasible solutions based on their sizes, i.e., the structural complexity in our work.
Thus, we construct two adaptions, namely~\textsc{EGS-Str} and~\textsc{EGS-StrDual},
to synthesize the queries under our problem setting.

\begin{itemize}[leftmargin=*]
	\item \textsc{EGS-Str} computes the longest substring of each string attribute in the positive tuple such that the string values of the attributes in negative ones do not contain it as the substring.
	We follow the priority function in \textsc{EGS}, which consists of the number of undesirable tuples eliminated per atomic constraint
	and the size of a query,
	to accelerate searching a query candidate with a small size.
	Finally, it obtains a query candidate for each positive tuple and groups the candidates as a result.
	
	\item \textsc{EGS-StrDual} further extends \textsc{EGS-Str} by considering the named entity coverage.
	Specifically, it prioritizes the refinable queries according to the three metrics, including the number of undesirable tuples eliminated per atomic constraint, the named entity coverage, and the size of a query.
	Other settings are the same as the ones of \textsc{EGS-Str}.
\end{itemize}

\begin{figure}[t]
	\centering
	\includegraphics[width=0.95\linewidth]{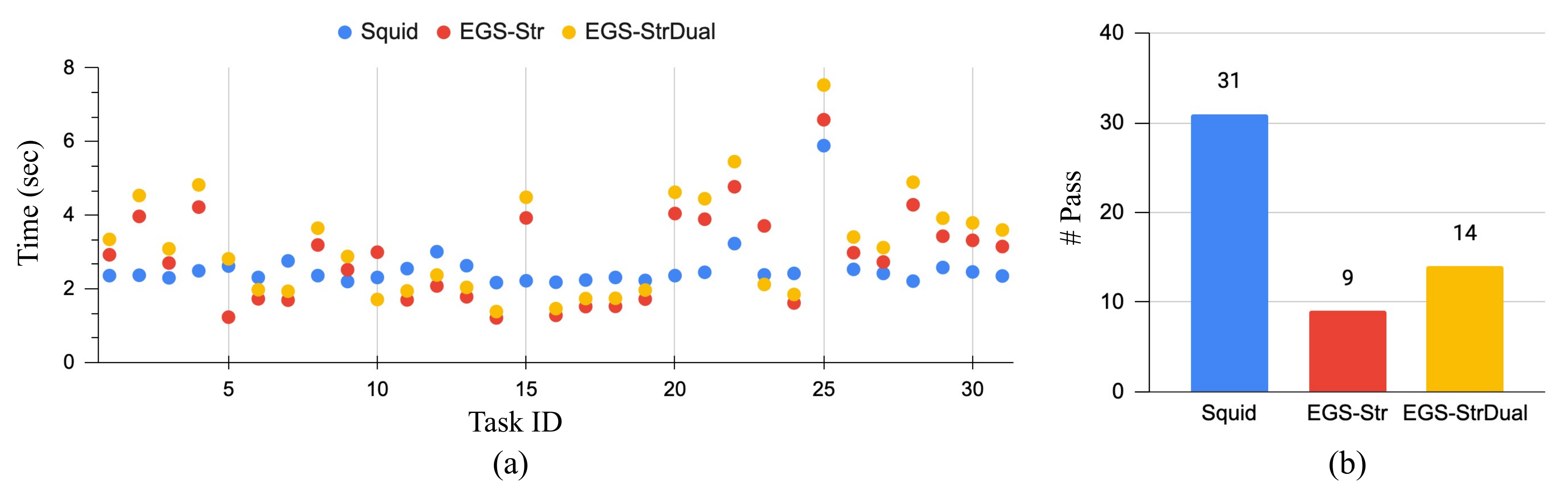}
	\vspace{-2mm}
	\caption{The time cost and the numbers of passed tasks of \ToolName, \textsc{EGS-Str}, and \textsc{EGS-StrDual}}
	\vspace{-3mm}
	\label{fig:RQ4}
\end{figure}

Figure~\ref{fig:RQ4} shows the results of the comparison.
On average, \textsc{EGS-Str} and \textsc{EGS-StrDual} spend 2.85 and 3.18 seconds on a synthesis instance, respectively,
while the average time cost of \ToolName\ is 2.56 seconds.
Although \textsc{EGS-Str} and \textsc{EGS-StrDual} accelerate the searching process with priority functions as the heuristic metrics,
they have to process the positive tuples in each round, 
and thus, the number of positive tuples can increase the overhead.

Meanwhile, the two baselines only succeed synthesizing queries for 9 and 14 tasks, respectively.
There are two root causes of their failures in synthesizing desired queries.
First, they synthesize the query candidates for each positive query separately 
and, thus, are more prone to over-fitting problems than \ToolName.
Second, the core algorithms of \textsc{EGS} and the two adaptions, which are the instantiations of inductive logic programming,
can not guarantee the obtained solutions are optimal under the given metrics.
As reported in~\cite{ThakkarNSANR21}, the query may not be of the minimal size if \textsc{EGS} leverages the number of undesirable tuples eliminated per atomic constraint to accelerate the searching process of a query candidate.
In our problem, our dual quantitative metrics increase the difficulty of achieving the optimal solutions with the two adaptions,
which causes the failures of the code search tasks.

\subsection{Discussion}

In what follows, we demonstrate the discussions on the limitations of \ToolName\
and several future works,
which can further improve the practicality of our techniques.

\smallskip
\emph{\textbf{Limitations.}}
Although \ToolName\ is demonstrated to be effective for code search,
it has two major limitations.
First, \ToolName\ cannot synthesize the query where the multiplicity of a relation is larger than the multiplicity bound $K$.
In other words, Theorem~\ref{thm:completeness} actually ensures partial completeness.
Although we may achieve completeness for all realizable instances 
by enumerating queries until obtaining a query candidate in Algorithm~\ref{alg:select},
\ToolName\ would fail to terminate for unrealizable instances.
Second,  \ToolName\ does not support synthesizing the queries with logical disjunctions.
However, when a code search task involves the matching of multiple patterns,
\ToolName\ would not discover the correct queries, which are out of the scope of the conjunctive queries.

\smallskip 
\emph{\textbf{Future Works.}}
In the future, we will attempt to propose an efficient decision procedure to identify unrealizable instances.
Equipped with the decision procedure,
we can only enumerate queries for realizable instances and generalize the query refinement by discarding the multiplicity bound.
Besides, it would be promising to generalize \ToolName\ for disjunctive query synthesis.
One possible adaption is to divide positive tuples into proper clusters and
synthesize a conjunctive part for each cluster separately,
following existing studies such as \textsc{EGS}~\cite{ThakkarNSANR21} and \textsc{RhoSynth}~\cite{Rule}.
In addition, we aim to expand \ToolName\ to diverse program domains, such as serverless functions~\cite{10.1145/3542929.3563471} and programs running on networking devices~\cite{10.1145/3386367.3431292}. 
These use cases have gained significant attention in recent years, which can pose new challenges on code search where new approaches may be needed.
Lastly, it would be meaningful to combine \ToolName\ with the techniques in the community of human-computer interaction~\cite{ChasinsGS21} to unleash its benefit for practice use.


\section{Related Work}
\ \ \ \ \emph{\textbf{Multi-modal Program Synthesis.}}
There has been a vast amount of literature on the multi-modal synthesis~\cite{ChenMF19, BaikJCJ20, ChenL0DBD21, Chen0YDD20, RazaGM15, GavranDM20, NaikMSWNR21}.
For example, the LTL formula synthesizer \textsc{LtlTalk}~\cite{GavranDM20}
maximizes the objective function that measures the similarity between the natural language description and the LTL formula,
and searches for the optimal solution that distinguishes the positive and negative examples.
\ToolName\ bears similarities to \textsc{LtlTalk} in terms of the prioritization,
while we use the named entities
to avoid the failure of semantic parsing of a sentence.
Another closely related work is a query synthesizer named \textsc{Sporq}~\cite{NaikMSWNR21}.
Based on code examples and user feedback,
\textsc{Sporq} iterates its PBE-based synthesis engine to refine the queries,
which demands verbose user interactions and a long time period.
In contrast, \ToolName\ automates the code search by solving a new multi-modal synthesis problem,
which only requires the users to specify code examples and a natural language description.
effectively relieving the user's burden in the searching process.

\smallskip
\emph{\textbf{Component-based Synthesis.}}
Several recent studies aim to compose several components (e.g., the classes and methods in the libraries)
into programs that achieve target functionalities~\cite{JhaGST10, GulwaniKT11, JamesGWDPJP20, PerelmanGBG12, GuoJJZWJP20, GveroKKP13}.
Typically, \textsc{SyPet}~\cite{FengM0DR17} and \textsc{APIphany}~\cite{GuoCT0SP22} both use the Petri net to encode the type signature of each function,
and collect the reachable paths to enumerate the well-typed sketches of the programs,
which prunes the search space at the start of the synthesis.
In our work, \ToolName\ leverages the schema graph to guide the enumerative search,
which share the similarity with existing studies.
However, our enumerative search space does not consist of the reachable paths in the schema graph,
and instead, contains different choices of selecting its nodes and edges.
Besides, unlike prior efforts~\cite{FengM0DR17, JamesGWDPJP20,GuoJJZWJP20, GuoCT0SP22}, \ToolName\ computes the activated relations and then discards unnecessary components,
i.e., dummy relations,
which distinguishes \ToolName\ significantly from other component-based synthesizers.

\smallskip
\emph{\textbf{Datalog Program Synthesis.}}
There have been many existing efforts of synthesizing Datalog programs~\cite{AlbarghouthiKNS17, SiLZAKN18, RaghothamanMZNS20, SiRHN19, MendelsonNRN21}.
For example, \textsc{Zaatar}~\cite{AlbarghouthiKNS17} encodes the input-output examples and Datalog programs with SMT formulas,
and synthesizes the candidate solution via constraint solving.
Unlike constraint-based approaches,
\textsc{ALPS}~\cite{SiLZAKN18} and \textsc{GENSYNTH}~\cite{MendelsonNRN21} synthesize 
target Datalog programs via the enumerative search,
which is similar to our synthesis algorithm.
However, existing studies do not tackle a large number of relations in the synthesis~\cite{AlbarghouthiKNS17, SiLZAKN18,  MendelsonNRN21}
or pursue an optimal solution with respect to a natural language description.
Meanwhile, they do not support the synthesis of string constraints, making their approaches incapable of string matching-based code search.
In contrast,
\ToolName\ ensures soundness, completeness, and optimality simultaneously
and synthesizes string constraints for string matching,
showing its potential in assisting real-world code search tasks.


\smallskip
\emph{\textbf{Datalog-based Program Analysis.}}
The past few decades have witnessed
the increasing popularity of Datalog-based program analysis~\cite{HajiyevVM06, WuZ021, YamaguchiGAR14, AvgustinovMJS16, SmaragdakisB10}.
For example,
\CodeQL\ encodes a program with a relational representation
and exposes a query language for query writing~\cite{AvgustinovMJS16}.
Several analyzers target more advanced semantic reasoning.
For example, the points-to and alias facts are depicted by two kinds of relations in \textsc{Doop}~\cite{BravenboerS09},
and meanwhile, pointer analysis algorithms are instantiated as Datalog rules~\cite{SmaragdakisB10}.
Other properties, such as def-use relation and type information,
can also be analyzed by existing analyzers~\cite{lam2005context, PashakhanlooN0D22}.
Our effort has shown the opportunity of unleashing the power of Datalog-based program analyzers seamlessly to support the code search automatically.



\section{Conclusion}

We propose an efficient synthesis algorithm \ToolName\ for a multi-modal conjunctive query synthesis problem,
which enables automatic code search using a Datalog-based program analyzer.
\ToolName\ reduces the search space via the representation reduction and the bounded refinement,
and meanwhile, conducts the query candidate selection with dual quantitative metrics.
It efficiently synthesizes the queries for \CaseNumber\ code search tasks 
with the guarantees of soundness, completeness, and optimality.
Its theoretical and empirical results offer strong evidence of its practical value in assisting code search in real-world scenarios.

\bibliography{sample-base}

\end{document}